\documentclass[10pt,twocolumn]{article}
\usepackage[margin=1in]{geometry}
\usepackage{times}

\usepackage{cite}
\usepackage{amsmath,amssymb,amsfonts}
\usepackage{algorithmic}
\usepackage{graphicx}
\usepackage{textcomp}
\usepackage{bm}
\usepackage{amsthm}
\usepackage[ruled,vlined]{algorithm2e}
\usepackage{color,soul}
\usepackage{xurl}
\usepackage{caption}
\usepackage{subcaption}
\usepackage{booktabs}
\usepackage{multirow}
\usepackage{enumitem}
\usepackage{tabularx}
\usepackage{pifont}
\usepackage{array}
\usepackage{hyperref}

\newtheorem{definition}{Definition}
\newtheorem{theorem}{Theorem}
\newtheorem{lemma}{Lemma}
\newtheorem{corollary}{Corollary}

\SetKwInput{KwInput}{Input}
\SetKwInput{KwOutput}{Output}
\SetKw{KwBy}{by}
\SetKw{Return}{return}

\newcommand{\cmark}{\ding{51}}%
\newcommand{\xmark}{\ding{55}}%

\title{\bfseries Versatile and Fast Location-Based Private Information Retrieval with Fully Homomorphic Encryption over the Torus}

\author{
Joon Soo Yoo$^{1}$, Taeho Kim$^{2}$, and Ji Won Yoon$^{1}$\\[0.5em]
\small $^1$School of Cybersecurity, Korea University, Seoul, Republic of Korea\\
\small $^2$Institute of ICT Planning and Evaluation (IITP), Daejeon, Republic of Korea
}

\usepackage{geometry}  
\begin{document}
\date{} 
\maketitle

\begin{abstract}
Location-based services often require users to share sensitive locational data, raising privacy concerns due to potential misuse or exploitation by untrusted servers. In response, we present \textsf{VeLoPIR}, a versatile location-based private information retrieval (PIR) system designed to preserve user privacy while enabling efficient and scalable query processing. \textsf{VeLoPIR} introduces three operational modes---interval validation, coordinate validation, and identifier matching---that support a broad range of real-world applications, including information and emergency alerts. To enhance performance, \textsf{VeLoPIR} incorporates multi-level algorithmic optimizations with parallel structures, achieving significant scalability across both CPU and GPU platforms. We also provide formal security and privacy proofs, confirming the system’s robustness under standard cryptographic assumptions. Extensive experiments on real-world datasets demonstrate that \textsf{VeLoPIR} achieves up to 11.55$\times$ speed-up over a prior baseline. The implementation of \textsf{VeLoPIR} is publicly available at \url{https://github.com/PrivStatBool/VeLoPIR}.
\end{abstract}


\section*{Acknowledgments}
This work was supported by the Institute of Information \& Communications Technology Planning \& Evaluation (IITP) grant funded by the Korean government (MSIT) (No. RS-2024-00460321, Development of Digital Asset Transaction Tracking Technology to Prevent Malicious Financial Conduct in the Digital Asset Market).

\section{Introduction}

In today’s digital landscape, location-based services (LBS) have become deeply embedded in everyday life—helping users discover nearby amenities, receive traffic updates, or get localized alerts. While these services offer convenience, they often rely on collecting and storing users’ precise geographic locations. This raises serious concerns about privacy. In many cases, the location data is gathered by large corporations or government agencies and may be repurposed or even sold to third parties, such as advertisers or data brokers.

Regulations such as the GDPR and national privacy laws are designed to prohibit unauthorized tracking of users’ location data. However, numerous investigations have revealed that major companies have violated these protections in practice~\cite{nyt2018,gravy2025}. Reports show that users’ location histories have been used to infer sensitive personal information, enabling long-term tracking, profiling, and even unauthorized sales to third parties. As a representative example, in 2025, Google reached a \$1.375 billion settlement with the state of Texas after being accused of secretly collecting users’ geolocation data, biometric identifiers, and incognito search history without consent—one of the largest privacy settlements in U.S. history~\cite{google2025}.

In conventional client-server communication, the reason the server can access user data is that both parties typically share a secret key established through a protocol such as Diffie-Hellman (DH)~\cite{diffie2022} or Elliptic Curve Diffie-Hellman (ECDH)~\cite{menezes2018}. While this ensures that data is encrypted during transmission, it does not prevent the server itself from decrypting and inspecting the data. As a result, although the server may correctly execute the user’s requested query, it retains full access to the plaintext, creating opportunities for misuse or unauthorized data collection. This inherent trust assumption poses a significant privacy risk in scenarios where the server cannot be fully trusted.

A promising cryptographic solution to prevent unauthorized access by the server—while still allowing it to process user queries—is homomorphic encryption (HE)~\cite{he}. Unlike traditional encryption methods, HE does not require a key exchange protocol such as Diffie-Hellman, because it allows computations to be performed directly on encrypted data without ever decrypting it. As a result, HE ensures data remains secure not only in transit but also from the server itself. In response to this capability, extensive research has been conducted on private information retrieval (PIR)~\cite{pir_gen, pir_pre_1, pir_pre_2, pir_sup, pir_sup2}, with a particular focus on minimizing computational overhead for both the client and server while ensuring that the server processes queries without learning their content.


Our work, \textsf{VeLoPIR}, falls within the family of PIR protocols, and more specifically, addresses the problem of location-based PIR. In this setting, a client can query information from a server based on their location, while ensuring that their geographic coordinates remain completely hidden from the server. \textsf{VeLoPIR} is built on \textsf{TFHE} (Fully Homomorphic Encryption over the Torus)~\cite{tfhe-1, tfhe-2}, a logic-gate-based FHE scheme well-suited for shallow~\footnote{In TFHE, “shallow circuits” refer to logic circuits with a relatively small number of sequential gates. For instance, \textsf{VeLoPIR} circuits typically have an estimated depth of around 40–80.}, non-linear circuits. Unlike multi-server~\cite{jain-two}, our single-server design ensures strong privacy guarantees while minimizing the client’s computational burden.

\textsf{VeLoPIR} introduces three operational modes—Interval Validation (\textsf{IntV}), Coordinate Validation (\textsf{CoV}), and Identifier Matching (\textsf{IdM})—each designed to support different types of location-based queries under strong privacy guarantees. While a simplified version of \textsf{IntV} has been explored in earlier work~\cite{yoo-torus}, \textsf{VeLoPIR} generalizes the model into a modular framework, extends it with additional query modes (\textsf{CoV} and \textsf{IdM}), and introduces systematic algorithmic and parallel optimizations for scalability. Importantly, \textsf{VeLoPIR} also provides formal security proofs and experimental validation using large-scale, real-world datasets, which were not addressed in prior designs.

\begin{figure}[htb]
    \centering
    \includegraphics[width=0.85\linewidth]{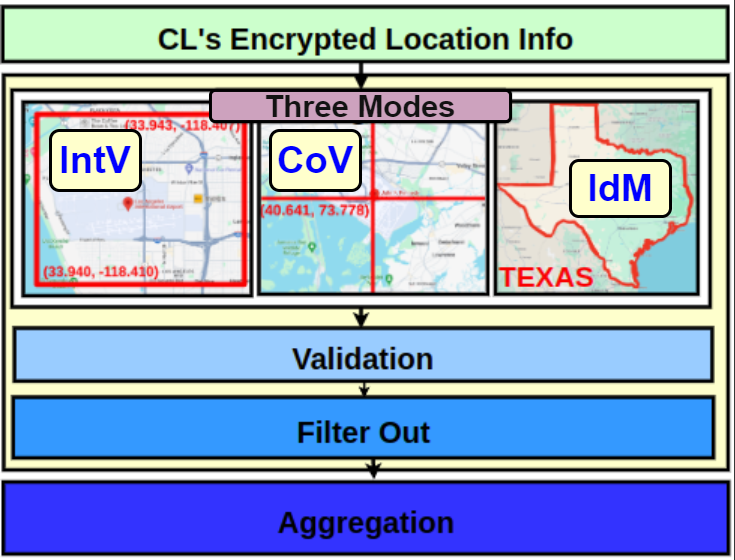}
    \caption{Overview of \textsf{VeLoPIR}'s Structure.}
    \label{fig:velo_overview}
\end{figure}

For a simplified illustration of \textsf{VeLoPIR}, the server processes a client’s (\textsf{CL}) encrypted locational data (see Fig.~\ref{fig:velo_overview}). At its core, \textsf{VeLoPIR} utilizes the operational modes \textsf{IntV}, \textsf{CoV}, and \textsf{IdM} as validation mechanisms for encrypted user locations, filtering out unrelated services. Each mode supports a distinct query type: \textsf{IntV} checks whether the location falls within a geographic interval, \textsf{CoV} validates exact coordinate matches, and \textsf{IdM} performs symbolic location matching based on identifiers such as city names or postal codes. The filtered result will be either the encryption of 0 or the encryption of the requested service, which is then aggregated to form the encrypted response. This response is returned to \textsf{CL}, who alone holds the secret key to decrypt the result. Throughout this process, the server has no access to the query result or any intermediate computations, thanks to FHE’s capability for computation on encrypted data.

In summary, the contributions of our work are as follows:
\begin{itemize}
    \item We introduce \textsf{VeLoPIR}, a versatile location-based PIR system that supports a broad range of real-world data types through three operational modes: \textsf{IntV}, \textsf{CoV}, and \textsf{IdM}, each preserving user locational privacy under different query structures.

    \item We design parallelized algorithms for \textsf{VeLoPIR} that exploit both CPU and GPU resources, enabling scalable and efficient evaluation over large datasets.

    \item We evaluate \textsf{VeLoPIR} using real-world datasets in both information and emergency alert scenarios, demonstrating significant improvements in query efficiency and practical applicability.

    \item We provide formal correctness proofs and conduct comprehensive security and privacy analysis, ensuring \textsf{VeLoPIR}'s robustness in protecting location privacy.

    \item Our system achieves up to 11.55$\times$ speed-up over baseline approaches for location-specific queries on the \texttt{covid-usa} dataset, reducing the query time from 113.44 seconds to 9.83 seconds.\footnote{The baseline is derived from the prior \textsf{LocPIR} framework~\cite{yoo-torus}.}
\end{itemize}

\subsection{Related Works}

\begin{table*}[!htb]
\centering
\caption{Comparison of Location-Based PIR Schemes (PQ Secure indicates post-quantum security).}
\begin{tabularx}{\textwidth}{|>{\centering\arraybackslash}X||c|c|c|c|c|c|}
\hline
\textbf{Scheme} & \textbf{Single Server} & \textbf{Offline (Client)} & \textbf{PQ Secure} & \textbf{Preprocessing} & \textbf{Location Specific} & \textbf{HE Scheme} \\ \hline
Lin et al.~\cite{lin-double} & \cmark & \xmark & \cmark & \cmark & \xmark & \textsf{BV} \\ \hline
Jain et al.~\cite{jain-two} & \xmark & \xmark & \xmark & \cmark & \xmark & \textsf{ElGamal/Paillier} \\ \hline
An et al.~\cite{an-covid} & \cmark & \cmark & \cmark & \cmark & \cmark & \textsf{BFV} \\ \hline
\textsf{VeLoPIR} \textbf{(Ours)} & \cmark & \xmark & \cmark & \xmark & \cmark & \textsf{TFHE} \\ \hline
\end{tabularx}
\label{tab:related_works_tab}
\end{table*}

Our work can be viewed as a specialized application of private information retrieval (PIR)~\cite{pir_gen}, focusing specifically on location-based queries. Traditional single-server PIR schemes often rely on server-side preprocessing~\cite{pir_pre_1, pir_pre_2} to achieve sublinear query time. However, such preprocessing is tightly coupled to the structure of the database and must be repeated whenever updates occur, limiting flexibility. Other approaches~\cite{pir_offline} introduce an offline phase to generate cryptographic hints that reduce the client's online cost, but these methods still incur significant overhead during setup. More importantly, existing PIR schemes generally do not address the unique challenges of location-based retrieval. We summarize relevant location-based PIR methods in Table~\ref{tab:related_works_tab}.

Lin et al.\cite{lin-double} focus on efficient PIR for large datasets using preprocessing and the BV scheme\cite{bv} from FHE. Their work primarily addresses the theoretical aspects of reducing computational complexity in general PIR, making it distinct from our \textsf{VeLoPIR}, which is optimized for location-specific queries in smaller, practical datasets.

Jain et al.~\cite{jain-two} propose a location-based recommendation service using a dual-serve model, leveraging a hybrid encryption approach with ElGamal~\cite{elgamal} and Paillier~\cite{paillier} schemes to preserve user privacy. Their system offloads the majority of computation to the two servers, thereby minimizing the client's computational burden. However, it relies on the assumption that the servers do not collude. As the dataset size increases, the recommendation time grows significantly, reaching 4355 seconds for 5000 elements.

An et al.~\cite{an-covid} address location-based private information retrieval in a single-server setup for COVID-19 alerts using the BFV scheme~\cite{bfv}. While their approach enables proximity calculations on encrypted data, it imposes offline computational overhead on the user's device. Moreover, the patient's locational privacy is not preserved, as raw location data is shared with the authorities, who can also access the patient's contact history, potentially exposing auxiliary information; and the overall computation time is approximately 399 seconds.


\section{Background} 
\label{sec:background}
We use bold uppercase letters to denote matrices (e.g., $\mathbf{S}$) and bold lowercase letters to indicate vectors (e.g., $\mathbf{s}$). Scalars are represented by italic letters (e.g., $s$). When referring to the binary representation of a scalar, we use bracket notation (e.g., $s[i]$), where $s[i]$ denotes the $i$-th bit of the binary representation of the scalar $s$. A summary of the notations used is provided in Appendix~\ref{appendix:notation} for reference.

\subsection{Fully Homomorphic Encryption (FHE)}

Homomorphic encryption enables computations on encrypted data without decryption, allowing a function $f$ to be evaluated on encrypted inputs. Given ciphertexts $\mathsf{Enc}_{\mathsf{sk}}(x)$ and $\mathsf{Enc}_{\mathsf{sk}}(y)$, the evaluation algorithm $\mathsf{Eval}$ produces an encrypted result that, when decrypted, matches $f(x, y)$:  
\begin{equation*}
    \mathsf{Dec}_\mathsf{sk}(\mathsf{Eval}_f(\mathsf{Enc}_\mathsf{sk}(x), \mathsf{Enc}_ \mathsf{sk}(y), \mathsf{evk})) = f(x, y).   
\end{equation*}
Only the holder of the secret key $\mathsf{sk}$ can decrypt the result, while the encrypted inputs remain secure under the hardness of cryptographic assumptions, such as lattice-based problems.

Several FHE schemes exist, including quantum-resistant options  NTRU~\cite{ntru} and Learning With Errors (LWE)-based schemes~\cite{lwe-2009, lwe-2}, both of which rely on the hardness of lattice problems. Although LWE-based schemes introduce noise that limits the depth of evaluations, this limitation is overcome by Gentry’s bootstrapping technique~\cite{gentry-fhe}, which reduces noise, enabling further evaluations on the ciphertext. In our work, we employ the \textsf{TFHE} scheme~\cite{tfhe-1}, an LWE-based homomorphic encryption approach optimized for fast bootstrapping and efficient homomorphic logic gate evaluation. While arithmetic-based schemes like CKKS~\cite{ckks} are effective for low-depth arithmetic circuits, \textsf{TFHE} excels in handling shallow circuits and nonlinear operations, making it particularly suited for our work.

\subsection{Torus Fully Homomorphic Encryption}

\textsf{TFHE} is an LWE-based encryption scheme that operates over the torus $\mathbb{T} = \mathbb{R}/\mathbb{Z}$. It uses multiple types of ciphertexts ($\mathsf{TLWE}$, $\mathsf{TRLWE}$, and $\mathsf{TRGSW}$) to enable efficient homomorphic operations, including the construction of logical gates. For the purpose of this paper, we specifically focus on the $\mathsf{TLWE}$ ciphertext $(\mathbf{a}, b)$. 

\begin{definition}[TLWE Problem]
Let \( n \in \mathbb{N} \) be a positive integer, and let \( \mathbf{s} = (s_1, \dots, s_n) \in \mathbb{B}^n \) be a secret vector where each \( s_i \) is sampled uniformly from the binary space \( \mathbb{B} = \{0, 1\} \). Let \( \chi \) be a Gaussian error distribution over the torus \( \mathbb{T} = \mathbb{R}/\mathbb{Z} \). The Torus Learning with Errors (TLWE) problem is the task of distinguishing between samples drawn from the following two distributions:
\begin{align*}
    \mathcal{D}_0 &= \{(\mathbf{a}, r) \mid \mathbf{a} \stackrel{\$}{\leftarrow} \mathbb{T}^n, r \stackrel{\$}{\leftarrow} \mathbb{T} \}, \\
    \mathcal{D}_1 &= \left\{(\mathbf{a}, b) \mid \mathbf{a} \stackrel{\$}{\leftarrow} \mathbb{T}^n, b = \langle \mathbf{a}, \mathbf{s} \rangle + e \mod 1, e \leftarrow \chi \right\}.
\end{align*}
\end{definition}

\begin{definition}[Advantage of an Adversary]
The advantage of an adversary \( \mathcal{A} \) in distinguishing between the two distributions \( \mathcal{D}_0 \) and \( \mathcal{D}_1 \) in the TLWE problem is defined as:
\begin{align*}
    \text{Adv}_{\text{TLWE}}^{\mathcal{A}} &= \bigg| \Pr[\mathcal{A}(\mathbf{a}, r) = 1 \mid (\mathbf{a}, r) \leftarrow \mathcal{D}_1] \\
    &\quad - \Pr[\mathcal{A}(\mathbf{a}, r) = 1 \mid (\mathbf{a}, r) \leftarrow \mathcal{D}_0] \bigg|.
\end{align*}
\end{definition}

The \textsf{TLWE} problem is said to be \( \lambda \)-secure if, for any probabilistic polynomial-time ($\mathsf{PPT}$) adversary \( \mathcal{A} \), the advantage \( \text{Adv}_{\text{TLWE}}^{\mathcal{A}} \) is at most \( 2^{-\lambda} \).

Based on the security of the $\mathsf{TLWE}$ ciphertext (similar to other $\mathsf{TRLWE}$ and $\mathsf{TRGSW}$ ciphertexts), the general $\mathsf{TLWE}$ encryption scheme is outlined as follows:
\begin{itemize}
    \item $\mathsf{KeyGen}(1^\lambda)$: Given a security parameter $\lambda$, the key generation algorithm defines $\mathsf{TLWE}$ parameters $n$ and $\sigma$, and outputs keys: a secret key $\mathbf{s} \stackrel{\$}{\leftarrow} \mathbb{B}^n$ and an evaluation key set $\mathsf{evk}$, which includes a bootstrapping key and a key switching key.
    \item $\mathsf{Enc}_\mathbf{s}(m)$: Given a binary message $m \in \mathbb{B}$, it is encoded to a plaintext $\mu \in \mathbb{T}$ using the function $\mathsf{Ecd}: m \mapsto m/4 - 1/8$. Next, the algorithm generates a masking vector $\mathbf{a}$ by uniformly sampling from the set $U(\mathbb{T}^n)$. The plaintext message $\mu$ is then encrypted using the secret key $\mathbf{s}$, resulting in a $\mathsf{TLWE}$ sample denoted as $\mathsf{ct} = (\mathbf{a}, b)$ where $b = \langle \mathbf{a}, \mathbf{s} \rangle + \mu + e$. Here, $e$ is a noise term drawn from a Gaussian distribution $N(0, \sigma)$.
    \item $\mathsf{Dec}_\mathbf{s}(\mathsf{ct})$: The decryption algorithm calculates the phase of the ciphertext $\mathsf{ct} = (\mathbf{a}, b)$ by $\varphi_{\mathbf{s}}(\mathsf{ct}) = b - \langle \mathbf{a}, \mathbf{s} \rangle$. The resulting phase $\mu + e$ is then rounded to the nearest plaintext message from the set $\{-1/8, 1/8\}$. Applying the inverse function of $\mathsf{Ecd}$ to the obtained plaintext message allows us to recover the original message bit $m \in \mathbb{B}$. Note that the error term $e$ must satisfy $|e| < 1/16$ to ensure correct decryption.
\end{itemize}

\begin{figure*}[htb]
    \centering
    
    \begin{subfigure}{0.48\textwidth}
        \centering
        \resizebox{\textwidth}{!}{ 
            \includegraphics{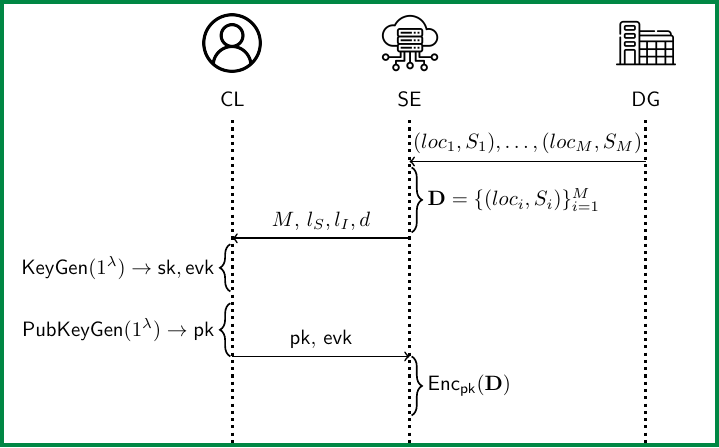} 
        }
        \caption{Preprocessing Phase}
        \label{fig:preprocessing}
    \end{subfigure}
    \hfill
    \begin{subfigure}{0.48\textwidth}
        \centering
        \resizebox{\textwidth}{!}{ 
            \includegraphics{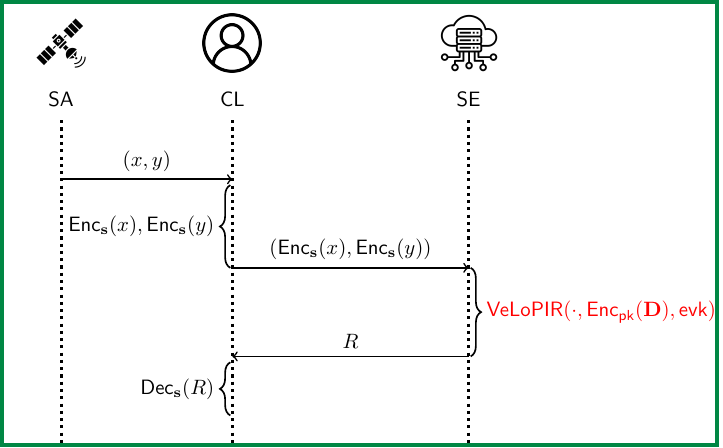} 
        }
        \caption{Evaluation Phase}
        \label{fig:evaluation}
    \end{subfigure}
    
    \caption{Preprocessing and Evaluation Phases of the \textsf{VeLoPIR} Protocol}
    \label{fig:VeLoPIR_phases}
\end{figure*}

\textsf{TFHE} logic gates are designed to refresh, or bootstrap, the noise with each logical operation. The bootstrapping ($\mathsf{Bootstrap}$) process in \textsf{TFHE} involves complex procedures such as $\mathsf{BlindRotate}$, $\mathsf{SampleExtract}$, and $\mathsf{KeySwitch}$. This bootstrapping phase typically consumes most of the time and memory during circuit evaluation, as it involves looking up elements in a precomputed $\mathsf{TRLWE}$ ciphertext table (see~\cite{tfhe-1, tfhe-2} for more details).

The construction of logical gates in \textsf{TFHE} utilizes the bootstrapping technique as described previously. Specifically, we demonstrate the homomorphic evaluation of the AND operation as follows:
\[
\mathsf{HomAND}(\mathsf{ct}_1, \mathsf{ct}_2, \mathsf{evk}) = \mathsf{Bootstrap}\left((\mathbf{0}, -\frac{1}{8}) + \mathsf{ct}_1 + \mathsf{ct}_2\right).
\]
Assuming that the magnitude of the errors in both \(\mathsf{ct}_1\) and \(\mathsf{ct}_2\) is less than \(1/16\), this procedure correctly produces the \(\mathsf{AND}\) result. Similarly, we use $\mathsf{HomXOR}$ and $\mathsf{HomXNOR}$ to denote homomorphic XOR and XNOR operations, respectively.

\section{Overall Protocol and Properties}

The proposed protocol involves four parties: a client (\textsf{CL}), a server (\textsf{SE}), a data generator (\textsf{DG}), and a satellite (\textsf{SA}). The primary objective is to protect the client’s locational data, ensuring that it can only be accessed by the client.

\subsection{Assumptions}

\noindent\textbf{Semi-Honest Server.} The server follows the protocol correctly but is considered curious and may attempt to learn the client's locational information \((x, y)\).

\noindent \textbf{Secure Client-Satellite Communication.} The client securely obtains its locational information \((x, y)\) from the satellite, ensuring confidentiality during transmission.

\subsection{Detailed Protocol}

\subsubsection{Preprocessing Phase}
The protocol involves three parties: \textsf{CL}, \textsf{SE}, and \textsf{DG} (see Fig.~\ref{fig:preprocessing}).

\begin{enumerate}
    \item \textbf{Data Collection.} \textsf{SE} collects data from \textsf{DG} to create a dataset $\mathbf{D} = \{({loc}_i, {S}_i)\}_{i=1}^{M}$. Here, each ${loc}_i$ represents a location, and ${S}_i$ represents associated data or services.

    \item \textbf{Secure Parameter Transmission.} \textsf{SE} sends the number of datasets $M$, data precision $l_S$, interval precision $l_I$, coordinate dimension $d$ to \textsf{CL} using a secure channel (e.g., TLS/SSL).

    \item \textbf{Key Generation by Client.} \textsf{CL} generates a secret key $\mathbf{s}$ and evaluation keys $\mathsf{evk}$ using the $\mathsf{KeyGen}(1^\lambda)$ algorithm. Additionally, \textsf{CL} performs the $\mathsf{PubKeyGen}$ algorithm to produce a public key set $\mathsf{pk}$ (refer to Algorithm~\ref{alg:PubKeyGen}). The generated public key set $\mathsf{pk}$ and $\mathsf{evk}$ are sent to \textsf{SE}.

    \begin{algorithm}
\caption{$\mathsf{PubKeyGen}(1^{\lambda}, \mathbf{s}, M, l_I, l_S, d) \rightarrow \mathsf{pk}$}
\label{alg:PubKeyGen}

Initialize $\mathsf{pk}^I, \mathsf{pk}^S \leftarrow \emptyset$\;

\For{$i = 1$ \KwTo $M$} {
    \For{$j \leftarrow 0$ \KwTo $l_{I}-1$} {
        \For{$k \leftarrow 1$ \KwTo $d$} {
            $\mathbf{a} \stackrel{\$}{\leftarrow} \mathbb{T}^n$; 
            $e \leftarrow \chi$; 
            $b \leftarrow \langle \mathbf{a}, \mathbf{s} \rangle + e$\;
            $\mathsf{ct}_{i,j,k} \leftarrow (\mathbf{a}, b)$\;
            $\mathsf{pk}^I \leftarrow \mathsf{pk}^I \cup \{\mathsf{ct}_{i,j,k}\}$\;    
        }
    }
    \For{$j \leftarrow 0$ \KwTo $l_S - 1$} {
        $\mathbf{a} \stackrel{\$}{\leftarrow} \mathbb{T}^n$; 
        $e \leftarrow \chi$; 
        $b \leftarrow \langle \mathbf{a}, \mathbf{s} \rangle + e$\;
        $\mathsf{ct}_{i,j} \leftarrow (\mathbf{a}, b)$\;
        $\mathsf{pk}^S \leftarrow \mathsf{pk}^S \cup \{\mathsf{ct}_{i,j}\}$\;
    }
}
\Return{$\mathsf{pk}= \{\mathsf{pk}^I, \mathsf{pk}^S\}$}\;
\end{algorithm}

    \item \textbf{Data Encryption by Server.} \textsf{SE} encrypts the dataset \(\mathbf{D}\) using \textsf{CL}'s public key \(\mathsf{pk}\) through the \textsf{ServerEnc} process (refer to Algorithm~\ref{alg:ServerEnc}). 
    
    \begin{algorithm}[htbp]
\caption{$\mathsf{ServerEnc}(\mathbf{D}, \mathsf{pk}) \rightarrow \mathsf{Enc}_{\mathsf{pk}}(\mathbf{D})$}
\label{alg:ServerEnc}

Initialize $\mathsf{Enc}_{\mathsf{pk}}(\mathbf{D}) \leftarrow \emptyset$\;

\For{$i \leftarrow 1$ \KwTo $M$} {
    \For{$j \leftarrow 0$ \KwTo $l_{I}-1$} {
        \For{$k \leftarrow 1$ \KwTo $d$} {
            $\mathsf{ct}_{i,j,k} \leftarrow (\mathbf{0}, \mathsf{Ecd}(x_{i,k}[j])) + \mathsf{pk}^I_{i,j,k}$\;
            $\mathsf{Enc}_{\mathsf{pk}}(\mathbf{D}) \leftarrow \mathsf{Enc}_{\mathsf{pk}}(\mathbf{D}) \cup \{\mathsf{ct}_{i,j,k}\}$\;    
        }
    }
    \For{$j \leftarrow 0$ \KwTo $l_S - 1$} {
        $\mathsf{ct}_{i,j} \leftarrow (\mathbf{0}, \mathsf{Ecd}(S_i[j])) + \mathsf{pk}^S_{i,j}$\;
        $\mathsf{Enc}_{\mathsf{pk}}(\mathbf{D}) \leftarrow \mathsf{Enc}_{\mathsf{pk}}(\mathbf{D}) \cup \{\mathsf{ct}_{i,j}\}$\; 
    }
}

\Return{$\mathsf{Enc}_{\mathsf{pk}}(\mathbf{D})$}\;

\end{algorithm}
\end{enumerate}



\subsubsection{Evaluation Phase} 
The evaluation protocol involves three parties: \textsf{SA}, \textsf{CL}, and \textsf{SE} (see Fig.~\ref{fig:evaluation}).

\begin{enumerate}
    \item \textbf{Location Reception by Client.} \textsf{CL} receives its locational information \((x, y)\) from \textsf{SA}, where \(x\) represents the latitude and \(y\) the longitude of \textsf{CL}'s location coordinates.
    
    \item \textbf{Encryption of Location.} \textsf{CL} encrypts \((x, y)\) using its secret key \(\mathbf{s}\) to obtain \((\mathsf{Enc}_{\mathbf{s}}(x), \mathsf{Enc}_{\mathbf{s}}(y))\). These encrypted coordinates are sent to \textsf{S}.
    
    \item \textbf{Evaluation by Server.} \textsf{SE} uses the evaluation key \(\mathsf{evk}\) to evaluate the \textsf{VeLoPIR} circuit given the encrypted GPS location \((\mathsf{Enc}_{\mathbf{s}}(x), \mathsf{Enc}_{\mathbf{s}}(y))\) of \textsf{CL} and its encrypted database \(\mathsf{Enc}_{\mathsf{pk}}(\mathbf{D})\). The goal is to obtain the encrypted result \(R\), which corresponds to the data associated with \textsf{CL}'s location. Algorithm~\ref{alg:VeLoPIR} in the following section provides the details of the server's evaluation process.
    
    \item \textbf{Return of Encrypted Result.} \textsf{SE} sends the encrypted result \(R\) back to \textsf{CL}.
    
    \item \textbf{Decryption by Client.} \textsf{CL} performs decryption using its secret key \(\mathbf{s}\) to obtain the result: \(\mathsf{Dec}_{\mathbf{s}}(R)\).
\end{enumerate}

\section{VeLoPIR Evaluation}

\textsf{VeLoPIR} is structured around three operational modes: \emph{Interval Validation}, \emph{Coordinate Validation}, and \emph{Identifier Matching}, with the latter being an extension of the coordinate-based approach.

\begin{itemize}
    \item \textbf{Interval Validation (\textsf{IntV})}: Checks if a client’s encrypted location falls within a specific geographical interval, such as city boundaries.
    \item \textbf{Coordinate Validation (\textsf{CoV})}: Verifies if the client’s encrypted coordinates match a predefined target location.
    \item \textbf{Identifier Matching (\textsf{IdM})}: Extends \textsf{CoV} by matching the client’s encrypted location to an identifier (e.g., city name or postal code) instead of direct coordinates.
\end{itemize}

Note that instead of directly matching coordinates \((x, y)\) (as in \textsf{CoV}), \textsf{IdM} maps these coordinates to a string identifier \(id_{\textsf{CL}}\). Since \textsf{IdM} follows the same validation process as \textsf{CoV}, the correctness, security, and privacy proofs for \textsf{CoV} naturally extend to \textsf{IdM}, ensuring that identifier matching offers the same guarantees as coordinate validation.

\subsection{Algorithm Overview}

For each entry in the encrypted database \(\mathsf{Enc}_{\mathsf{pk}}(\mathbf{D})\), which contains encrypted location coordinates \({loc}_i\) and associated data \({S}_i\) for up to \(M\) entries, the following steps are performed:

\begin{enumerate}
    \item \textbf{Validation.} Determine whether the user's encrypted coordinates \((\mathsf{Enc}_{\mathbf{s}}(x), \mathsf{Enc}_{\mathbf{s}}(y))\) fall within \({loc}_i\). If a match is found, output \(\mathsf{Enc}_{\mathbf{s}}(1)\); otherwise, output \(\mathsf{Enc}_{\mathbf{s}}(0)\). Store the validation result as \(v\). This step utilizes \textsf{IntV}, \textsf{CoV}, or \textsf{IdM}.

    \item \textbf{Zero Out Unrelated Data.} Use the validation result \(v\) to eliminate unrelated services or data by performing a bitwise AND operation with all \({S}_i\). If validation is successful at position \(\kappa\), only \({S}_\kappa\) remains unchanged, while other \({S}_i\) become vectors of encrypted zeros \(\mathsf{Enc}_{\mathbf{s}}(0)\).
\end{enumerate}

\noindent\textbf{Aggregate Results.} After the iteration of all locational entries, sum all \({S}_i\). Since only \({S}_\kappa\) is non-zero, the result will be \({S}_\kappa\).

\begin{algorithm}[htbp]
\caption{\(\mathsf{VeLoPIR}(\mathsf{Enc}_{\mathbf{s}}(x), \mathsf{Enc}_{\mathbf{s}}(y), \mathsf{Enc}_{\mathsf{pk}}(\mathbf{D}), \mathsf{evk})\)}
\label{alg:VeLoPIR}

\For{$i \leftarrow 1$ \KwTo $M$} {
    \textbf{Validation:} Check if \((\mathsf{Enc}_{\mathbf{s}}(x), \mathsf{Enc}_{\mathbf{s}}(y))\) fall within or match \({loc}_i\)\;
    \tcc{\text{Use $\mathsf{IntV}$, $\mathsf{CoV}$, or $\mathsf{IdM}$}}

    \If{validation is successful}{
        $v \leftarrow \mathsf{Enc}_{\mathbf{s}}(1)$\;
    } \Else {
        $v \leftarrow \mathsf{Enc}_{\mathbf{s}}(0)$\;
    }
    
    \textbf{Zero Out Unrelated Data:} Apply $v$ to filter the data \({S}_{\kappa}\)\;
    \ForEach{${S}_i$}{
        ${S}_i \leftarrow \mathsf{HomBitwiseAND}(v, \mathsf{Enc}_{\mathbf{s}}({S}_i), \mathsf{evk})$\;
    }
}
\textbf{Aggregation:} $R \leftarrow \mathsf{HomSum}(\{\mathsf{Enc}_{\mathbf{s}}(S_i)\}_{i=1}^M, \mathsf{evk})$\;

\Return{$R$}\;
\end{algorithm}

\subsection{Core Algorithms}

\noindent\textbf{(1) Interval Validation, \textsf{IntV}.}
For interval validation, the locational coordinates \(loc_i\) are defined by intervals of latitude \((x_{{left}}, x_{{right}})\) and longitude \((y_{{left}}, y_{{right}})\). The algorithm leverages nonlinear comparison operations, specifically less than or equal ($\mathsf{HomCompLE}$) and less than ($\mathsf{HomCompL}$), to perform these comparisons homomorphically (refer to Algorithm~\ref{alg:InterValid}). 

\begin{algorithm}[htbp]
\caption{\(\mathsf{IntV}(\mathsf{Enc}_{\mathbf{s}}(x), \mathsf{Enc}_{\mathbf{s}}(y), \mathsf{Enc}_{\mathbf{s}}(loc), \mathsf{evk})\)}
\label{alg:InterValid}

\tcc{$\mathsf{Enc}_{\mathbf{s}}(loc) = (\mathsf{Enc}_{\mathbf{s}}(x_{left}), \dots, \mathsf{Enc}_{\mathbf{s}}(y_{right}))$}

\(v_{x_{left}} \leftarrow \mathsf{HomCompLE}(\mathsf{Enc}_{\mathbf{s}}(x_{left}), \mathsf{Enc}_{\mathbf{s}}(x), l_I, \mathsf{evk})\)\;
\(v_{x_{right}} \leftarrow \mathsf{HomCompL}(\mathsf{Enc}_{\mathbf{s}}(x), \mathsf{Enc}_{\mathbf{s}}(x_{right}), l_I, \mathsf{evk})\)\;

\(v_{y_{left}} \leftarrow \mathsf{HomCompLE}(\mathsf{Enc}_{\mathbf{s}}(y_{left}), \mathsf{Enc}_{\mathbf{s}}(y), l_I, \mathsf{evk})\)\;
\(v_{y_{right}} \leftarrow \mathsf{HomCompL}(\mathsf{Enc}_{\mathbf{s}}(y), \mathsf{Enc}_{\mathbf{s}}(y_{right}), l_I, \mathsf{evk})\)\;

\(v_{x} \leftarrow \mathsf{HomAND}(v_{x_{left}}, v_{x_{right}}, \mathsf{evk})\)\;

\(v_{y} \leftarrow \mathsf{HomAND}(v_{y_{left}}, v_{y_{right}}, \mathsf{evk})\)\;

\(v \leftarrow \mathsf{HomAND}(v_x, v_y, \mathsf{evk})\)\;

\Return \(v\)\;

\end{algorithm}

Given the encrypted client’s locational information $(\mathsf{Enc}_{\mathbf{s}}(x),$ $ \mathsf{Enc}_{\mathbf{s}}(y))$, the algorithm \textsf{IntV} evaluates the following plaintext condition: \(x_{{left}} \leq x < x_{{right}}\) and \(y_{{left}} \leq y < y_{{right}}\). The validation outputs \(\mathsf{Enc}_{\mathbf{s}}(1)\) if the conditions are satisfied; otherwise, it outputs \(\mathsf{Enc}_{\mathbf{s}}(0)\).

Note that $\mathsf{HomAND}$ denotes the AND operation performed homomorphically on encrypted bits, producing the AND result. Similarly, $\mathsf{HomXOR}$ and $\mathsf{HomXNOR}$ represent homomorphic gates that perform XOR and XNOR operations, respectively (refer to Section~\ref{sec:background} for additional background details). 

We provide an in-depth discussion of the underlying algorithms, $\mathsf{HomCompL}$ and $\mathsf{HomCompLE}$, which form the backbone of $\mathsf{IntV}$, followed by the correctness of $\mathsf{IntV}$.


\noindent\textbf{Backbones of {IntV}: FHE Comparisons.}
$\mathsf{HomCompLE}$ is a nonlinear FHE comparison circuit that takes two ciphertexts \(\mathsf{ct_1}\) and \(\mathsf{ct_2}\), and outputs \(\mathsf{Enc}_{\mathbf{s}}(1)\) if \(x \leq y\), and \(\mathsf{Enc}_{\mathbf{s}}(0)\) otherwise. This operation is carried out using the evaluation key \(\mathsf{evk}\), with \(\mathsf{ct_1}\) and \(\mathsf{ct_2}\) being encryptions of \(z_1\) and \(z_2\) under the secret key \(\mathbf{s}\) (refer to Algorithm~\ref{alg:HomCompLE}). 

\begin{algorithm}[htbp]
\caption{\(\mathsf{HomCompLE}(\mathsf{ct_1}, \mathsf{ct_2}, l_I, \mathsf{evk})\)}
\label{alg:HomCompLE}

\(t_0 \leftarrow \mathsf{HomXOR}(\mathsf{ct_1}[l_I - 1], \mathsf{ct_2}[l_I - 1], \mathsf{evk})\)\;

$t_2 \leftarrow \mathsf{Enc}_{\mathbf{s}}^{\mathsf{TLWE}}(\mathbf{0}, \frac{1}{8})$ 

\For{$i \leftarrow 0$ \KwTo $l_I - 2$} {
    \(t_1 \leftarrow \mathsf{HomXNOR}(\mathsf{ct_1}[i], \mathsf{ct_2}[i], \mathsf{evk})\)\;
    \(t_2 \leftarrow \mathsf{HomMUX}(t_1, t_2, \mathsf{ct_2}[i], \mathsf{evk})\)\;
}

\(r \leftarrow \mathsf{HomMUX}(t_0, \mathsf{ct_1}[l_I - 1], t_2, \mathsf{evk})\)

\Return \(r\)

\end{algorithm}

The circuit's design is fundamental to the correctness of the $\mathsf{IntV}$ algorithm, as discussed in~\cite{yoo-torus}. For completeness, Theorem~\ref{thm:comple} formally establishes the correctness of $\mathsf{HomCompLE}$, with the proof provided in Appendix~\ref{appendix:proof_LE}. 

\begin{theorem}
\label{thm:comple}
The homomorphic comparison $\mathsf{HomCompLE}(\mathsf{ct_1},$ $ \mathsf{ct_2},$ $ \mathsf{evk})$ stated in Algorithm~\ref{alg:HomCompLE} correctly evaluates whether $z_1 \leq z_2$ and outputs the result $r$ as:
\begin{equation*}
r = 
\begin{cases} 
\mathsf{Enc}_{\mathbf{s}}(1), & \text{if } z_1 \leq z_2 \\
\mathsf{Enc}_{\mathbf{s}}(0), & \text{otherwise}
\end{cases}
\end{equation*}
where $\mathsf{ct_1}$ and $\mathsf{ct_2}$ are encryptions of $z_1$ and $z_2$, respectively, under the secret key $\mathbf{s}$.
\end{theorem}
\begin{proof}
    See Appendix~\ref{appendix:proof_LE}.
\end{proof}

$\mathsf{HomCompL}$ is similar in design to $\mathsf{HomCompLE}$, with the only variation being the initialization of the \(t_2\) variable, which is set to a trivial \(\mathsf{TLWE}\) ciphertext of 0 (refer to Algorithm~\ref{alg:HomCompL}).

\begin{algorithm}[htbp]
\caption{\(\mathsf{HomCompL}(\mathsf{ct_1}, \mathsf{ct_2}, l_I, \mathsf{evk})\)}
\label{alg:HomCompL}

\(t_0 \leftarrow \mathsf{HomXOR}(\mathsf{ct_1}[l_I - 1], \mathsf{ct_2}[l_I - 1], \mathsf{evk})\)\;

\(t_2 \leftarrow \mathsf{Enc}_{\mathbf{s}}^{\mathsf{TLWE}}(\mathbf{0}, -\frac{1}{8})\) 

\For{$i \leftarrow 0$ \KwTo $l_I - 2$} {
    \(t_1 \leftarrow \mathsf{HomXNOR}(\mathsf{ct_1}[i], \mathsf{ct_2}[i], \mathsf{evk})\)\;
    \(t_2 \leftarrow \mathsf{HomMUX}(t_1, t_2, \mathsf{ct_2}[i], \mathsf{evk})\)\;
}

\(r \leftarrow \mathsf{HomMUX}(t_0, \mathsf{ct_1}[l_I - 1], t_2, \mathsf{evk})\)\;

\Return \(r\)\;

\end{algorithm}

Corollary~\ref{cor:compl} establishes the correctness of Algorithm~\ref{alg:HomCompL}.

\begin{corollary}
\label{cor:compl}
The homomorphic comparison $\mathsf{HomCompL}(\mathsf{ct_1},$ $ \mathsf{ct_2}, $ $\mathsf{evk})$ stated in Algorithm~\ref{alg:HomCompL} correctly evaluates whether \(z_1 < z_2\) and outputs the result \(r\) as:
\begin{equation*}
r = 
\begin{cases} 
\mathsf{Enc}_{\mathbf{s}}(1), & \text{if } z_1 < z_2 \\
\mathsf{Enc}_{\mathbf{s}}(0), & \text{otherwise}
\end{cases}
\end{equation*}
where \(\mathsf{ct_1}\) and \(\mathsf{ct_2}\) are encryptions of \(z_1\) and \(z_2\), respectively, under the secret key \(\mathbf{s}\).
\end{corollary}

\begin{proof}
    See Appendix~\ref{appendix:proof_L}.
\end{proof}

We now formally establish the correctness of $\mathsf{IntV}$, as illustrated in Algorithm~\ref{alg:InterValid}, with the following theorem.

\begin{theorem}
\label{thm:bb1}
The interval validation algorithm (\textsf{IntV}) as stated in Algorithm~\ref{alg:InterValid} correctly evaluates whether the encrypted coordinate \(\mathsf{Enc}_{\mathbf{s}}(x)\), \(\mathsf{Enc}_{\mathbf{s}}(y)\) lies within the interval $loc = (x_{left}, x_{right}, $ $y_{left}, y_{right})$. Specifically, it outputs:
\[
v \leftarrow \mathsf{Enc}_{\mathbf{s}}(1)
\]
if \(x_{left} \leq x < x_{right}\) and \(y_{left} \leq y < y_{right}\), and otherwise outputs $\mathsf{Enc}_{\mathbf{s}}(0)$.

\end{theorem}

\begin{proof}
    See Appendix~\ref{appendix:proof_bb1}.
\end{proof}


\noindent \textbf{(2) Coordinate Validation, \textsf{CoV} (and \textsf{IdM}).}  
\textsf{CoV} efficiently matches the client’s location coordinates \((x, y)\) with a predefined service location \({loc}_i\). If the client’s encrypted coordinates match \({loc}_i\), the corresponding service \({S}_i\) is returned. \textsf{CoV} relies on the homomorphic equality comparison algorithm \(\mathsf{HomEQ}\) (refer to Algorithm~\ref{alg:HomEQ}) for secure matching. Since \textsf{IdM} extends \textsf{CoV} by matching coordinates to an identifier, the correctness naturally follows from \textsf{CoV}, and we therefore focus only on \textsf{CoV}.

\begin{algorithm}[htbp]
\caption{\(\mathsf{HomEQ}(\mathsf{ct_1}, \mathsf{ct_2}, l_I, \mathsf{evk})\)}
\label{alg:HomEQ}

$r \leftarrow \mathsf{Enc}_{\mathbf{s}}^{\mathsf{TLWE}}(\mathbf{0}, \frac{1}{8})$ 

\For{$i \leftarrow 0$ \KwTo $l_I - 1$} {
    \(t \leftarrow \mathsf{HomXNOR}(\mathsf{ct_1}[i], \mathsf{ct_2}[i], \mathsf{evk})\)\;
    \(r \leftarrow \mathsf{HomAND}(r, t, \mathsf{evk})\)\;
}

\Return \(r\)

\end{algorithm}

The correctness of \textsf{CoV} (and \textsf{IdM}) relies on \textsf{HomEQ}, as established in Lemma~\ref{lem:EQ}, with the proof in Appendix~\ref{appendix:proof_EQ}.

\begin{lemma}
\label{lem:EQ}
The homomorphic equality comparison $\mathsf{HomEQ}($ $\mathsf{ct_1},$ $ \mathsf{ct_2},$ $ \mathsf{evk})$ correctly evaluates whether $z_1 = z_2$ and outputs the result $r$ as:
\begin{equation*}
r = 
\begin{cases} 
\mathsf{Enc}_{\mathbf{s}}(1), & \text{if } z_1 = z_2 \\
\mathsf{Enc}_{\mathbf{s}}(0), & \text{otherwise}
\end{cases}
\end{equation*}
where $\mathsf{ct_1}$ and $\mathsf{ct_2}$ are encryptions of $z_1$ and $z_2$, respectively, under the secret key $\mathbf{s}$.
\end{lemma}

\begin{proof}
    See Appendix~\ref{appendix:proof_EQ}
\end{proof}

\textsf{CoV} and \textsf{IdM} both utilize the \(\mathsf{HomEQ}\) operation, as shown in Algorithm~\ref{alg:BB2}. In \textsf{IdM}, instead of applying homomorphic AND operation, \(\mathsf{HomEQ}\) directly evaluates the encrypted identifier \(\mathsf{Enc}_{\mathbf{s}}(id_{\mathsf{CL}})\).

\begin{algorithm}[htbp]
\caption{\(\mathsf{CoV}(\mathsf{Enc}_{\mathbf{s}}(x), \mathsf{Enc}_{\mathbf{s}}(y), \mathsf{Enc}_{\mathbf{s}}(loc), \mathsf{evk})\)}
\label{alg:BB2}

\tcc{($loc = (loc_x, loc_y)$)}

\(v_x \leftarrow \mathsf{HomEQ}(\mathsf{Enc}_{\mathbf{s}}(x), \mathsf{Enc}_{\mathbf{s}}(loc_x), \mathsf{evk})\)\;

\(v_y \leftarrow \mathsf{HomEQ}(\mathsf{Enc}_{\mathbf{s}}(y), \mathsf{Enc}_{\mathbf{s}}(loc_y), \mathsf{evk})\)\;

\(v \leftarrow \mathsf{HomAND}(v_x, v_y, \mathsf{evk})\)\;

\Return \(v\)\;

\end{algorithm}

The correctness of \textsf{CoV} (and \textsf{IdM}) is shown in Theorem~\ref{thm:bb2}.

\begin{theorem}
\label{thm:bb2}
The coordinate validation algorithm (\textsf{CoV}, similarly \textsf{IdM}), as stated in Algorithm~\ref{alg:BB2}, correctly evaluates whether the encrypted coordinates \(\mathsf{Enc}_{\mathbf{s}}(x)\) and \(\mathsf{Enc}_{\mathbf{s}}(y)\) match the encrypted location \(\mathsf{Enc}_{\mathbf{s}}(loc)\). Specifically, it outputs:
\[
v \leftarrow \mathsf{Enc}_{\mathbf{s}}(1)
\]
if \(x = loc_x\) and \(y = loc_y\), and otherwise outputs \(\mathsf{Enc}_{\mathbf{s}}(0)\).
\end{theorem}

\begin{proof}
The correctness of the \textsf{CoV} algorithm can be directly inferred from the homomorphic equality checks performed by \(\mathsf{HomEQ}\). Specifically, the algorithm uses \(\mathsf{HomEQ}\) to compare the encrypted coordinates \(\mathsf{Enc}_{\mathbf{s}}(x)\) and \(\mathsf{Enc}_{\mathbf{s}}(y)\) with \(\mathsf{Enc}_{\mathbf{s}}(loc_x)\) and \(\mathsf{Enc}_{\mathbf{s}}(loc_y)\) respectively.

For the final output \(v\) to be \(\mathsf{Enc}_{\mathbf{s}}(1)\), both comparisons must return \(\mathsf{Enc}_{\mathbf{s}}(1)\). This means that \(x\) must equal \(loc_x\) and \(y\) must equal \(loc_y\). If either comparison fails, the result will be \(\mathsf{Enc}_{\mathbf{s}}(0)\) due to the \(\mathsf{HomAND}\) operation. 
\end{proof}

\subsection{Supporting Algorithms}
\noindent\textbf{(1) Zero-out Function.}
Once the validation result \(v\) is obtained, indicating whether the validation was successful (\(\mathsf{Enc}_{\mathbf{s}}(1)\)) or not (\(\mathsf{Enc}_{\mathbf{s}}(0)\)), this encrypted result can be used to zero out the service data \(S_i\). The function \(\mathsf{HomBitwiseAND}(v, \mathsf{Enc}_{\mathbf{s}}(S_i), \mathsf{evk})\) is applied to produce \(\mathsf{Enc}_{\mathbf{s}}(S_{\kappa})\), where \(S_{\kappa}\) is the desired locational information associated with \({loc}_{\kappa}\) at index \(\kappa\) (refer to \textsf{VeLoPIR} Algorithm~\ref{alg:VeLoPIR}). 

\begin{algorithm}[htbp]
\caption{\(\mathsf{HomBitwiseAND}(v, \mathsf{ct}, \mathsf{evk})\)}
\label{alg:HomBitwiseAND}

\tcc{\(l\) is the length of \(\mathsf{ct}\)}
\For{$i \leftarrow 1$ \KwTo $l$}{
    \(\mathsf{ct}[i] \leftarrow \mathsf{HomAND}(v, \mathsf{ct}[i], \mathsf{evk})\)\;
}
\Return \(\mathsf{ct}\)\;

\end{algorithm}

The \(\mathsf{HomBitwiseAND}\) function ensures that \(\mathsf{Enc}_{\mathbf{s}}(S_i)\) remains non-zero only when \(v\) is \(\mathsf{Enc}_{\mathbf{s}}(1)\), achieved by applying the \(\mathsf{HomAND}\) gate bitwise (see Algorithm~\ref{alg:HomBitwiseAND}). If \(v\) is \(\mathsf{Enc}_{\mathbf{s}}(0)\), then \(\mathsf{Enc}_{\mathbf{s}}(S_i)\) is zeroed out, meaning all bits are \(\mathsf{Enc}_{\mathbf{s}}(0)\) due to the \(\mathsf{HomAND}\) operation. This is summarized in Lemma~\ref{lem:bitAND}, with the proof omitted.

\begin{lemma}
\label{lem:bitAND}
The homomorphic function \(\mathsf{HomBitwiseAND}\) in Algorithm~\ref{alg:HomBitwiseAND} outputs \(\mathsf{ct}\) if \(v = \mathsf{Enc}_{\mathbf{s}}(1)\), and outputs \(\mathsf{Enc}_{\mathbf{s}}(0)\) for all bits of \(\mathsf{ct}\) if \(v = \mathsf{Enc}_{\mathbf{s}}(0)\). Given the encrypted single bit \(v\), the function effectively zeroes out the encrypted vector \(\mathsf{ct}\) based on the value of \(v\).
\end{lemma}

\noindent\textbf{(2) Aggregation.}
At the final stage of the \textsf{VeLoPIR} algorithm, we aggregate the results of \(\mathsf{Enc}_{\mathbf{s}}(S_i)\) that have been validated by the encrypted bit \(v\). Assuming no \( \text{loc}_i \) intersects with another \( \text{loc}_j \) for \( i \neq j \), the resulting set \( \{\mathsf{Enc}_{\mathbf{s}}(S_i)\}_{i=1}^{M} \) will contain only one non-zero element, \( \mathsf{Enc}_{\mathbf{s}}(S_{\kappa}) \), while the other \( \mathsf{Enc}_{\mathbf{s}}(S_i) \) values will be \( \mathsf{Enc}_{\mathbf{s}}(0) \) across all \( l_S \) bits.

\begin{algorithm}[htbp]
\caption{\(\mathsf{HomSum}(\{\mathsf{Enc}_{\mathbf{s}}(S_i)\}_{i=1}^{M}, \mathsf{evk})\)}
\label{alg:HomSum}

\tcc{Initialize \(S\) as a vector of encrypted zeros of length \(l_S\)}
\(S \leftarrow [\mathsf{Enc}_{\mathbf{s}}(0), \dots, \mathsf{Enc}_{\mathbf{s}}(0)]\)

\For{$i \leftarrow 1$ \KwTo $M$}{
    \For{$j \leftarrow 0$ \KwTo $l_S - 1$}{
        \(S[j] \leftarrow \mathsf{HomXOR}(S[j], \mathsf{Enc}_{\mathbf{s}}(S_i[j]), \mathsf{evk})\)\;
    }
}
\Return \(S\)\;

\end{algorithm}

Since all \( \mathsf{Enc}_{\mathbf{s}}(S_i) \) values are encrypted, the server cannot distinguish between the encryption of zeros and \( \mathsf{Enc}_{\mathbf{s}}(S_{\kappa}) \). Therefore, the server cannot naively select \( S_{\kappa} \) but must instead consider all possible cases. This is where homomorphic XORing comes into play through the \(\mathsf{HomSum}\) algorithm, which ensures that the correct \( S_{\kappa} \) is obtained without revealing any unnecessary information (see Algorithm~\ref{alg:HomSum}). The proof of Lemma~\ref{lem:HomSum} explains the correctness and necessity of the \(\mathsf{HomSum}\) operation.

\begin{lemma}
\label{lem:HomSum}
The homomorphic summation \(\mathsf{HomSum}\) in Algorithm~\ref{alg:HomSum} over a set \(\{\mathsf{Enc}_{\mathbf{s}}(S_i)\}_{i=1}^{M}\) correctly outputs \(\mathsf{Enc}_{\mathbf{s}}(S_\kappa)\), where the set \(\{\mathsf{Enc}_{\mathbf{s}}(S_i)\}_{i=1}^{M}\) contains only one non-zero \(\mathsf{Enc}_{\mathbf{s}}(S_\kappa)\), and the rest \(\mathsf{Enc}_{\mathbf{s}}(S_i)\) are encryptions of zeros across all \(l_S\) bits.
\end{lemma}

\begin{proof}
Given that all \(\mathsf{Enc}_{\mathbf{s}}(S_i)\) except \(\mathsf{Enc}_{\mathbf{s}}(S_\kappa)\) are encryptions of zero across all \(l_S\) bits, the homomorphic XOR operation effectively acts as an addition. Specifically, the iterative operation \(S[j] \leftarrow \mathsf{HomXOR}(S[j], \mathsf{Enc}_{\mathbf{s}}(S_i[j]), \mathsf{evk})\) simplifies to:
\[
S[j] \leftarrow \mathsf{HomXOR}(\mathsf{Enc}_{\mathbf{s}}(0), \mathsf{Enc}_{\mathbf{s}}(S_\kappa[j]), \mathsf{evk})
\]
Since XORing any bit with 0 yields the original bit, the operation correctly outputs \(\mathsf{Enc}_{\mathbf{s}}(S_\kappa)\) for \(S\) in Algorithm~\ref{alg:HomSum}.
\end{proof}

\subsection{Correctness of VeLoPIR}

We can now conclude that the \textsf{VeLoPIR} circuit correctly evaluates and retrieves the desired locational information \(R = \mathsf{Enc}_{\mathbf{s}}(S_\kappa)\), corresponding to the encrypted coordinates \((\mathsf{Enc}_{\mathbf{s}}(x), \mathsf{Enc}_{\mathbf{s}}(y))\), as established through the series of preceding theorems and lemmas, culminating in Theorem~\ref{thm:VeLoPIR}.

\begin{theorem}
\label{thm:VeLoPIR}
The homomorphic location-based information retrieval algorithm, \(\mathsf{VeLoPIR}\), as described in Algorithm~\ref{alg:VeLoPIR}, correctly evaluates and retrieves the encrypted service data \(\mathsf{Enc}_{\mathbf{s}}(S_\kappa)\) corresponding to the encrypted locational coordinates \((\mathsf{Enc}_{\mathbf{s}}(x), \mathsf{Enc}_{\mathbf{s}}(y))\) using the evaluation key \(\mathsf{evk}\).
\end{theorem}

\begin{proof}
First, by Theorems~\ref{thm:bb1} and \ref{thm:bb2}, the validation step in \(\mathsf{VeLoPIR}\) correctly evaluates whether the encrypted coordinates \(\mathsf{Enc}_{\mathbf{s}}(x)\) and \(\mathsf{Enc}_{\mathbf{s}}(y)\) match or fall within the location \(loc_\kappa\). The result of this validation is an encrypted bit \(v\) that equals \(\mathsf{Enc}_{\mathbf{s}}(1)\) if the coordinates match and \(\mathsf{Enc}_{\mathbf{s}}(0)\) otherwise. Next, using Lemma~\ref{lem:bitAND}, the \(\mathsf{HomBitwiseAND}\) function applies this validation result \(v\) to each encrypted service data \(\mathsf{Enc}_{\mathbf{s}}(S_i)\). This operation zeroes out all \(\mathsf{Enc}_{\mathbf{s}}(S_i)\) except for \(\mathsf{Enc}_{\mathbf{s}}(S_\kappa)\), which corresponds to the validated location \(loc_\kappa\). Finally, by Lemma~\ref{lem:HomSum}, the \(\mathsf{HomSum}\) operation correctly aggregates the results to produce \(R = \mathsf{Enc}_{\mathbf{s}}(S_\kappa)\). This ensures that only the service data corresponding to the validated location is returned.

\end{proof}
\section{Security and Privacy Analysis}
\noindent\textbf{Security and Privacy of VeLoPIR.}
We propose that the \textsf{VeLoPIR} system is both secure and private. This is because all operations are performed on encrypted data, with no decryption occurring during the evaluation process. The following theorem formalizes the security and privacy of our system.

\begin{theorem}
Assuming that the size of the database \(M\) and the precision parameters \(l_I\) and \(l_S\) are bounded by \(O(\lambda^k)\) for some constant \(k\), the proposed \textsf{VeLoPIR} system is \(\lambda\)-\textbf{secure} and preserves the \textbf{privacy} of the client's location.
\end{theorem}

\begin{proof}
\noindent{\emph{(Security)}} Given that the size of the database \(M\) and the precision parameters \(l_I\) and \(l_S\) are bounded by \(O(\lambda^k)\) for some constant \(k\), the adversary's advantage \(\text{Adv}_{\text{TLWE}_{1 \text{ to } W}}^{\mathcal{A}}\) remains negligible:
\[
\text{Adv}_{\text{TLWE}_{1 \text{ to } W_s}}^{\mathcal{A}} = \sum_{i=1}^{W} \text{Adv}_{\text{TLWE}_i}^{\mathcal{A}} < \frac{O(W_s)}{2^\lambda}
\]
where \(W_s = O(M l_I l_S)\) denotes the number of $\mathsf{TLWE}$ samples occurring throughout the \textsf{VeLoPIR} system. Therefore, the system is \(\lambda\)-secure.

\noindent{\emph{(Privacy)}} Since the input locational information $(\mathsf{Enc}_{\mathbf{s}}(x), $ $\mathsf{Enc}_{\mathbf{s}}(y)) $of the client is encrypted, and all intermediate results as well as the final output \(R = \mathsf{Enc}_{\mathbf{s}}(S_\kappa)\) are encrypted under the client's secret key, no information about the client is leaked except for auxiliary information such as \(M\), \(l_S\), \(l_I\), and \(d\). Thus, the privacy of the user's location is preserved.

\end{proof}

\noindent \textbf{Threat 1: Auxiliary Information Leakage.}
One potential threat to the \textsf{VeLoPIR} protocol arises if the server obtains auxiliary information. This information could include the number of data entries $M$, the service precision $l_S$, the interval precision $l_I$, and the coordinate dimension $d$ from the client data. However, even if the adversary (which could be the server or a third party intercepting the communication between the client and server) gains access to these auxiliary $\mathsf{TLWE}$ samples, the adversary's ability to break the $\mathsf{TLWE}$ encryption and retrieve the secret key $\mathbf{s}$ remains negligible.

\begin{theorem}
\label{thm:aux_adv}
The advantage of an adversary $\mathcal{A}$, who has access to auxiliary public keys $\mathsf{pk}$ and information such as $M$, $l_S$, $l_I$, and $d$, is negligible. Specifically, for a polynomial number of samples $W = M  l_I d + M  l_S$, the advantage is bounded by:
\[
\text{Adv}_{\textsf{TLWE}_{1 \text{ to } W}}^{\mathcal{A}} < \frac{O(W)}{2^\lambda}
\]
where $W = M  l_I d + M  l_S$.
\end{theorem}

\begin{proof}
Given that the $\mathsf{TLWE}$ sample is $\lambda$-secure, the adversary's advantage for each individual sample is bounded by $2^{-\lambda}$. The total advantage across $W$ samples is:
\[
\text{Adv}_{\textsf{TLWE}_{1 \text{ to } W}}^{\mathcal{A}} = \sum_{i=1}^{W} \text{Adv}_{\textsf{TLWE}_i}^{\mathcal{A}} < \frac{O(W)}{2^\lambda}
\]
Since $W = M  l_I d + M  l_S$ grows polynomially with the input size, the adversary's total advantage remains negligible, ensuring the security of the protocol.
\end{proof}

\noindent\textbf{Threat 2: Server Access to Raw Data.}
Another potential threat arises when the server has raw access to the dataset $\mathbf{D} = \{({loc}_i, {S}_i)\}_{i=1}^{M}$, which contains service data associated with specific locational coordinates. This access is necessary for the server to provide the relevant services to the client, as it aggregates this information from local data generators. However, the primary objective of the \textsf{VeLoPIR} protocol is to protect the user's location, not the service data $S_i$ itself (which the server already possesses). Consequently, the ability of the adversary (in this case, the server) to solve the $\mathsf{TLWE}$ sample and obtain the user's secret key $\mathbf{s}$ remains negligible. Lemma~\ref{cor:raw_data} directly follows from Theorem~\ref{thm:aux_adv}.

\begin{corollary}
\label{cor:raw_data}
The advantage of the server in the \textsf{VeLoPIR} protocol, given access to the service data $\mathbf{D} = \{({loc}_i, {S}_i)\}_{i=1}^{M}$, is negligible:
\[
\text{Adv}_{\textsf{TLWE}_{1 \text{ to } W}}^{\mathcal{A}} < \frac{O(W)}{2^\lambda}
\]
where $W = M  l_I d + M  l_S$, as defined in Theorem~\ref{thm:aux_adv}.
\end{corollary}

\begin{proof}
    See Appendix~\ref{appendix:proof_raw}.
\end{proof}

\noindent\textbf{Threat 3: Insecure Communication Channel from Satellite to Client.}
A potential threat to the \textsf{VeLoPIR} protocol is the vulnerability of GPS information during transmission (refer to Figure~\ref{fig:evaluation}) from the satellite (\textsf{SA}) to the client (\textsf{Cl}), such as in GPS spoofing~\cite{spoof-1, spoof-2} and jamming attacks~\cite{jam-1, jam-2}. These attacks could alter the user's location before it is encrypted and processed by \textsf{VeLoPIR}, potentially rendering the retrieved information inaccurate. However, this threat is beyond the scope of the \textsf{VeLoPIR} protocol, which focuses on protecting the user's location data post-encryption. \textsf{VeLoPIR} ensures that location-based queries are securely processed without revealing the user's location to the server or any other party during the evaluation process.

Securing the communication channel between the satellite and the client is an ongoing research focus, particularly within the realm of Post-Quantum Cryptography (PQC). Efforts are being made to establish a secure network using PQC algorithms~\cite{pqc-survey, pqc-cyber, pqc-iot}, including Lattice-based cryptographic schemes like those based on the LWE problem. Given that the \textsf{VeLoPIR} protocol is also built on LWE, implementing LWE-based PQC schemes for secure satellite-to-client communication would not only enhance security but also ensure seamless integration with \textsf{VeLoPIR}.
\section{VeLoPIR Optimization}
\label{sec:optimization}

\noindent \textbf{Bit-Level Parallel Acceleration.} In logic-based FHE schemes like \textsf{TFHE}, homomorphic gate operations, such as $\mathsf{HomAND}$, are significantly slower than their plaintext equivalents. For example, evaluating a homomorphic AND gate, which involves steps like $\mathsf{BlindRotate}$, $\mathsf{SampleExtract}$, and $\mathsf{KeySwitch}$, takes around 13 ms, whereas a plaintext gate takes only 0.3 ns—making the homomorphic operation about 40 million times slower. This considerable overhead makes \textsf{TFHE} well-suited for parallel optimization, as CPU and GPU parallelism can substantially improve performance. We leverage this observation to optimize the core algorithms in \textsf{VeLoPIR}.

\subsection{Core Functionality Optimizations}

The \textsf{VeLoPIR} circuit can be optimized at multiple levels, ranging from bitwise operations to high-level structural enhancements. We categorize these into three levels of parallel optimization---\textbf{Innermost}, \textbf{Mid-Level}, and \textbf{Outermost}---corresponding to the hierarchical structure of computation within the system (refer to Table~\ref{tab:hardware_complexity} for a summary of these optimization levels and the associated components).


\noindent \textbf{(1) Innermost Optimization: Nonlinear HE Operations.}
Optimizing nonlinear HE operations like \(\mathsf{HomCompLE}\), \(\mathsf{HomCompL}\), and \(\mathsf{HomEQ}\) can significantly enhance the performance of the \textsf{VeLoPIR} system. 

\noindent \emph{Homomorphic Comparsions.} In the \(\mathsf{HomCompLE}\) algorithm (Algorithm~\ref{alg:HomCompLE}), the following loop is key:
\begin{align*}
    t_1 &\leftarrow \mathsf{HomXNOR}(\mathsf{ct_1}[i], \mathsf{ct_2}[i], \mathsf{evk}) \\
    t_2 &\leftarrow \mathsf{HomMUX}(t_1, t_2, \mathsf{ct_2}[i], \mathsf{evk})
\end{align*}
Here, \(t_2\) is sequentially updated based on the result of the previous \(t_1\) value. This sequential dependency means that the \(t_2\) updates cannot be parallelized. However, the \(\mathsf{HomXNOR}\) operations for each bit position are independent and can be executed in parallel using parallel processing units (see Algorithm~\ref{alg:HomCompLeOPT}).

\begin{algorithm}[htbp]
\caption{\(\mathsf{HomCompLeOPT}(\mathsf{ct_1}, \mathsf{ct_2}, \mathsf{evk}, n_p)\)}
\label{alg:HomCompLeOPT}

\(t_0 \leftarrow \mathsf{HomXOR}(\mathsf{ct_1}[l_I - 1], \mathsf{ct_2}[l_I - 1], \mathsf{evk})\)\;

\(t_1 \leftarrow \mathsf{Enc}_{\mathbf{s}}^{\mathsf{TLWE}}(\mathbf{0}, \frac{1}{8})\)\;

\ForEach{ $i \leftarrow 0$ \KwTo $l_I - 2$}{
    \tcc{allocate parallel processing units}
    \(t_{\text{XNOR}}[i] \leftarrow \mathsf{HomXNOR}(\mathsf{ct_1}[i], \mathsf{ct_2}[i], \mathsf{evk})\)\;
}

\For{$i \leftarrow 0$ \KwTo $l_I - 2$} {
    \(t_1 \leftarrow \mathsf{HomMUX}(t_{\text{XNOR}}[i], t_1, \mathsf{ct_2}[i], \mathsf{evk})\)\;
}

\(r \leftarrow \mathsf{HomMUX}(t_0, \mathsf{ct_1}[l_I - 1], t_1, \mathsf{evk})\)\;

\Return \(r\)\;

\end{algorithm}

Similarly, \(\mathsf{HomCompL}\) can be optimized in the same manner as \(\mathsf{HomCompLE}\), improving performance by parallelizing the independent \(\mathsf{HomXNOR}\) operations. Both \(\mathsf{HomCompL}\) and \(\mathsf{HomCompLE}\) utilize parallel resources to their maximum capacity, up to \( l_I \) units.

\noindent \emph{Homomorphic Equality.} The \(\mathsf{HomEQ}\) algorithm contains no dependencies between the outcomes, allowing parallel computation of \(\mathsf{HomXNOR}\) on the input ciphertexts \(\mathsf{ct_1}[i]\) and \(\mathsf{ct_2}[i]\) up to \(l_I\), producing \(t_1[i]\) values. These values can then be reduced in parallel using pairwise \(\mathsf{HomAND}\) operations, reaching a depth of \(\log(l_I)\). The \(\mathsf{HomEqOPT}\) algorithm can utilize up to \(l_I\) parallel processing units.

\begin{algorithm}[htbp]
\caption{\(\mathsf{HomEqOPT}(\mathsf{ct_1}, \mathsf{ct_2}, \mathsf{evk}, n_p)\)}
\label{alg:HomEquiOPT}

\ForEach{ $i \leftarrow 0$ \KwTo $l_I - 1$}{
    \tcc{allocate parallel processing units}
    \(t_{\text{1}}[i] \leftarrow \mathsf{HomXNOR}(\mathsf{ct_1}[i], \mathsf{ct_2}[i], \mathsf{evk})\)\;
}

\tcc{Parallel reduction}
\For{\( k \leftarrow 1\) \KwTo \( l_I\) \KwBy $\times2$}{
    \For{$i \leftarrow 0$ \KwTo $l_I - 1$ \KwBy $2k$}{
        \If{$i + k < l_I$}{
            \(t_{\text{1}}[i] \leftarrow \mathsf{HomAND}(t_{\text{1}}[i], t_{\text{1}}[i + k], \mathsf{evk})\)\;
        }
    }
}

\Return \(r \leftarrow t_{\text{1}}[0]\)\;

\end{algorithm}

\noindent \textbf{(2) Mid-Level Optimization: \(\mathsf{HomBitwiseAND}\).}
A straightforward optimization approach involves allocating parallel resources to the $\mathsf{HomAND}$ gate evaluation in Algorithm~\ref{alg:HomBitwiseAND}:
\begin{equation}
\mathsf{ct}[i] \leftarrow \mathsf{HomAND}(v, \mathsf{ct}[i], \mathsf{evk}).
\end{equation}
Since VeLoPIR requires the evaluation of $\mathsf{HomBitwiseAND}(v$ $, \mathsf{Enc}_{\mathbf{s}}(S_i),$ $ \mathsf{evk})$, this operation necessitates $l_S$ parallel resources for maximum parallelization.

\noindent \textbf{(3) Mid-Level Optimization: Validation Modules (\textsf{IntV}, \textsf{CoV}).} 
\noindent \emph{Interval Validation.} Validation in \textsf{IntV} can be optimized by parallelizing the computation of independent variables \(v_{x_{left}}, v_{x_{right}}, v_{y_{left}},\) and \(v_{y_{right}}\). These variables, which utilize similar operations to \(\mathsf{HomCompLE}\) and \(\mathsf{HomCompL}\), can be computed simultaneously. We can parallel process the following operations line by line:
\begin{align*}
v_{x_{left}} &\leftarrow \mathsf{HomCompLE}(\mathsf{Enc}_{\mathbf{s}}(x_{left}), \mathsf{Enc}_{\mathbf{s}}(x), \mathsf{evk})\; \\
v_{x_{right}} &\leftarrow \mathsf{HomCompL}(\mathsf{Enc}_{\mathbf{s}}(x), \mathsf{Enc}_{\mathbf{s}}(x_{right}), \mathsf{evk})\; \\
v_{y_{left}} &\leftarrow \mathsf{HomCompLE}(\mathsf{Enc}_{\mathbf{s}}(y_{left}), \mathsf{Enc}_{\mathbf{s}}(y), \mathsf{evk})\; \\
v_{y_{right}} &\leftarrow \mathsf{HomCompL}(\mathsf{Enc}_{\mathbf{s}}(y), \mathsf{Enc}_{\mathbf{s}}(y_{right}), \mathsf{evk})\; 
\end{align*}
Additionally, \(v_x\) and \(v_y\) can be computed independently afterward. The algorithm can be parallelized using up to \(2d\) resources.

\noindent \emph{Coordinate Validation.} 
\textsf{CoV} can also benefit from parallel processing by pairing the independent \(v_x\) and \(v_y\) operations, both of which use the \(\mathsf{HomEQ}\) function. This allows the use of up to \(d\) parallel processing units.

\noindent \textbf{(4) Outermost Optimization: Large-Scale Evaluation.}
The \textsf{VeLoPIR} evaluation can be optimized by leveraging \(M\) parallel processors to handle the large-scale processing unit, specifically the evaluation of each \({S}_i\) in the \textsf{VeLoPIR} circuit. Since each \({S}_i\) operates independently and requires similar computational effort, parallelization ensures efficient processing. The optimization can be designed as follows:

\tcc{allocate parallel processing units}
\ForEach{ $i \leftarrow 0$ \KwTo $M$}{
    
    \textbf{Validation $v$ $\&$ Zero Out Unrelated Data:}\; 
    
    \ForEach{${S}_i$}{
        ${S}_i \leftarrow \mathsf{HomBitwiseAND}(v, \mathsf{Enc}_{\mathbf{s}}({S}_i), \mathsf{evk})$\;
    }
}

\noindent \textbf{(5) Outermost Optimization: Aggregation via HomSum.}
The \(\mathsf{HomSum}\) algorithm can be optimized using two approaches: parallel reduction and Bit-AND optimization. First, we apply pairwise evaluation of the \(\mathsf{HomXOR}\) operation on the encrypted service data \(\mathsf{Enc}_{\mathbf{s}}(S_i[j])\) using parallel reduction across \(M\) processing units. This reduces the number of operations logarithmically with respect to \(M\).

Next, for each pairwise \(\mathsf{HomXOR}\) evaluation, we apply Bit-AND optimization, utilizing up to \(l_S\) parallel units for the operation:
\begin{equation*}
    \mathsf{Enc}_{\mathbf{s}}(S_i[j]) \leftarrow \mathsf{HomXOR}(\mathsf{Enc}_{\mathbf{s}}(S_i[j]), \mathsf{Enc}_{\mathbf{s}}(S_{i+k}[j]), \mathsf{evk})
\end{equation*}
where \(k\) increases from 1 to \(\log_2(M)\) during pairwise evaluation. Thus, the maximum parallel processing units required for \(\mathsf{HomSum}\) optimization is \( (M \cdot l_S) /2\).

\section{Theoretical Resource Complexity}

\noindent \textbf{Parallel Processing Resource Complexity.} In \textsf{VeLoPIR}, we utilize parallel processing to enhance algorithm efficiency. The resource complexity is computed by considering the hierarchical structure of the algorithms, such as \(\mathsf{HomCompLE}\) and \(\mathsf{HomCompL}\) within \textsf{IntV}. Table~\ref{tab:hardware_complexity} outlines the maximum number of parallel units required for each component. The overall parallel processing resource complexity simplifies to:
\[
O(M \cdot (l_S + d \cdot l_I))
\]
This accounts for the dominant factors affecting parallel execution across the algorithm’s components.

\begin{table}[htb!]
\centering
\caption{Parallel resource complexity and hierarchy in \textsf{VeLoPIR}. Levels refer to the scope of parallelism: \textbf{Outermost} (across records), \textbf{Mid-Level} (across validation blocks), and \textbf{Innermost} (within comparison primitives).}
\begin{tabularx}{\linewidth}{Xcc}
\toprule
\textbf{Algorithm}            & \textbf{Processing Units ($n_p$)}  & \textbf{Parallel Level} \\ 
\midrule
\(\mathsf{Large\text{-}Scale}\)                & \( M \)                          & Outermost \\
\(\mathsf{HomSum}\)                            & \( \frac{M l_S}{2} \)            & Outermost \\ 
\(\mathsf{HomBitwiseAND}\)                     & \( l_S \)                        & Mid-Level \\
\textsf{IntV}   & \( 2d \)                         & Mid-Level \\
\textsf{CoV} / \textsf{IdM}                    & \( d \, (1) \)                   & Mid-Level \\
\(\mathsf{HomCompLE} \, / \, \mathsf{HomCompL}\) & \( l_I \)                      & Innermost \\
\(\mathsf{HomEQ}\)                              & \( l_I \)                        & Innermost \\
\bottomrule
\end{tabularx}
\label{tab:hardware_complexity}
\end{table}

\noindent \textbf{Theoretical Speed-up Bound.}
Theoretically, the maximum possible speed-up for \textsf{VeLoPIR} is \(O(M \cdot (l_S + d \cdot l_I))\) as with the parallel processing resource complexity. This bound assumes that all parts of the algorithm can be fully parallelized without overhead. However, as we will show in the evaluation section, the actual performance is affected by factors such as communication overhead between the GPU and CPU. This limits the real-world speed-up compared to the theoretical maximum.

\section{Application Scenarios}

We explore three privacy-preserving strategies for private information alerts and location-based services: \textsf{IntV}, \textsf{CoV}, and \textsf{IdM}.

\noindent \textbf{(1) Bounding Box Validation (\textsf{IntV})}. Bounding box validation is an efficient strategy when the user’s location needs to be verified within a certain geographic boundary without revealing their exact position. This is particularly relevant in scenarios like pandemic disease alerts, where the client is interested in receiving alerts about nearby confirmed cases without explicitly disclosing their location. We demonstrate the performance of the \textsf{IntV} approach using datasets such as \texttt{covid-usa} and \texttt{covid-kor}.

\noindent \textbf{(2) Exact Coordinate Validation (\textsf{CoV})}. In certain cases, the user needs highly precise information about their exact geographic position. This approach is essential for applications like crime alerts or disaster management, where precise location-based notifications are crucial for immediate response. We use the \texttt{gdacs} dataset to demonstrate the efficiency of \textsf{CoV} in these situations.

\noindent \textbf{(3) Identifier Matching (\textsf{IdM})}. In many modern datasets, especially those related to healthcare or meteorological data, information is often categorized by administrative regions such as cities or states. In these cases, validating a user’s location based on an identifier (e.g., city or state name) can be more efficient and relevant than using geographic coordinates. \textsf{IdM} leverages this structure to provide efficient and privacy-preserving location matching. We showcase the performance of \textsf{IdM} with the \texttt{covid-usa}, \texttt{covid-kor}, and \texttt{weather-us} datasets.

\section{Evaluation}
We evaluate \textsf{VeLoPIR} with the following key questions:
\begin{itemize}
    \item \textbf{RQ1. Parallelization.} Can the core components of \textsf{VeLoPIR}---including outermost, mid-level, and innermost levels of parallelism---be efficiently parallelized across CPU cores and GPUs? Which hardware platform (CPU or GPU) is optimal for the \textsf{VeLoPIR} circuit?

    \item \textbf{RQ2. Efficiency.} How much speed-up can be achieved through our core optimizations across different datasets?

    \item \textbf{RQ3. Applicability.} Which operational mode (\textsf{IntV}, \textsf{CoV}, or \textsf{IdM}) is best suited for each real-world application scenario?
\end{itemize}

\noindent \textbf{Reference Benchmark.} In this work, we use \textsf{LocPIR}~\cite{yoo-torus} as a reference benchmark for comparison. We evaluate \textsf{VeLoPIR} to demonstrate its extended functionality and optimized performance for practical applications. All evaluations are conducted using real-world datasets (summarized in Table~\ref{tab:dataset_summary}) to ensure that the improvements in \textsf{VeLoPIR} are measured against practical, relevant data.

\subsection{Environment Setup}

\noindent\textbf{Hardware.} Experiments were conducted on a 13th Gen Intel Core i9-13900K processor (24 cores, 32 threads, 5.8 GHz max frequency). The system operated on Ubuntu 24.04 LTS, utilizing \textsf{TFHE} library version 1.1. For GPU parallel processing, the system was equipped with an NVIDIA GeForce RTX 4060 Ti GPU (16 GB GDDR6 memory) running CUDA version 12.4. 

\noindent \textsf{TFHE} \textbf{Parameters.} Our model employed a standard security level of 128-bit ($\lambda_{128}$) as in~\cite{tfhe-2} (for details, see Appendix~\ref{appendix:tfhe_params}).

\noindent\textbf{Dataset (refer to Tab.~\ref{tab:dataset_summary}).} The \texttt{covid-kor} dataset~\cite{covid-kor}, provided by the Korea Disease Control and Prevention Agency (KDCA), contains daily records of COVID-19 patient incidences across nine major cities in Korea as of October 26, 2021.
The \texttt{covid-usa} dataset~\cite{covid-us} contains daily confirmed COVID-19 cases across U.S. states, recorded on February 11, 2023.
The \texttt{gdacs} dataset~\cite{gdacs} is sourced from the Global Disaster Alert and Coordination System (GDACS), capturing disaster information (e.g., earthquakes, tsunamis, etc.) for hazard assessment.
Lastly, the \texttt{weather-usa} dataset~\cite{weather-usa} includes daily weather reports for U.S. cities, specifically from January 3, 2016. 


\begin{table}[!h]
\centering
\caption{Summary of Real World Datasets}
\resizebox{\linewidth}{!}{  
\begin{tabular}{|c|c|c|c|c|}
\hline
\textbf{Dataset} & $M$ & $l_I$ & $l_S$ & \textbf{Description} \\ \hline
\texttt{covid-kor~\cite{covid-kor}} & 9 & 16 & 9 & COVID-19 patients (Korea) \\ \hline
\texttt{covid-usa~\cite{covid-us}} & 58 & 16 & 22 & COVID-19 patients (US) \\ \hline
\texttt{gdacs~\cite{gdacs}} & 9 & 16 & 16 & Disaster Alert (Global) \\ \hline
\texttt{weather-usa~\cite{weather-usa}} & 304 & 16 & 128 & Weather data (US) \\ \hline
\end{tabular}
}
\label{tab:dataset_summary}
\end{table}

\subsection{Core Functionality Optimizations (RQ1)}

    

    

\begin{figure*}[!htb]
    \centering
    
    \begin{subfigure}{0.24\textwidth}
        \centering
        \includegraphics[width=\textwidth]{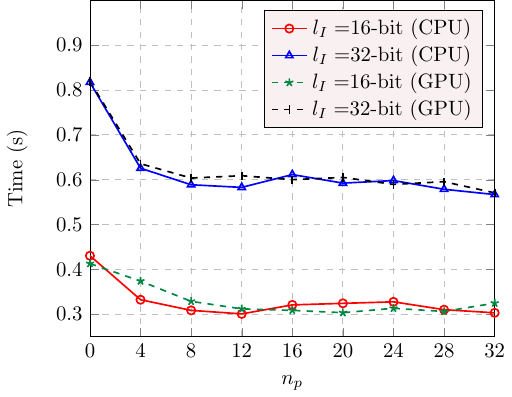}
        \caption{$\mathsf{HomCompLE}$}
        \label{fig:comple}
    \end{subfigure}
    \hfill
    \begin{subfigure}{0.24\textwidth}
        \centering
        \includegraphics[width=\textwidth]{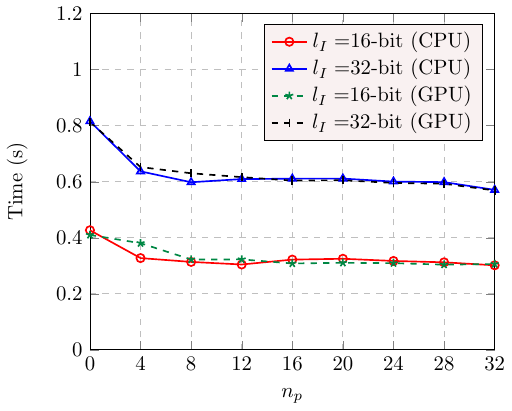}
        \caption{$\mathsf{HomCompL}$}
        \label{fig:compl}
    \end{subfigure}
    \hfill
    \begin{subfigure}{0.24\textwidth}
        \centering
        \includegraphics[width=\textwidth]{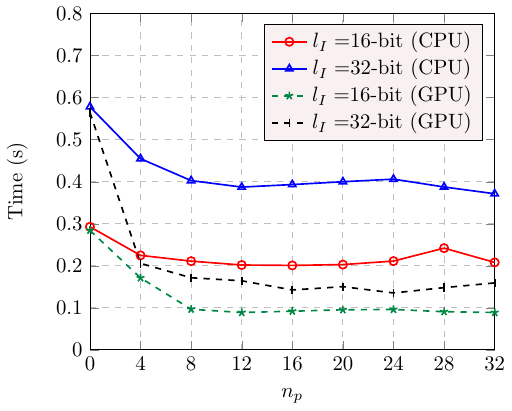}
        \caption{$\mathsf{HomEQ}$}
        \label{fig:equi}
    \end{subfigure}
    \hfill
    \begin{subfigure}{0.24\textwidth}
        \centering
        \includegraphics[width=\textwidth]{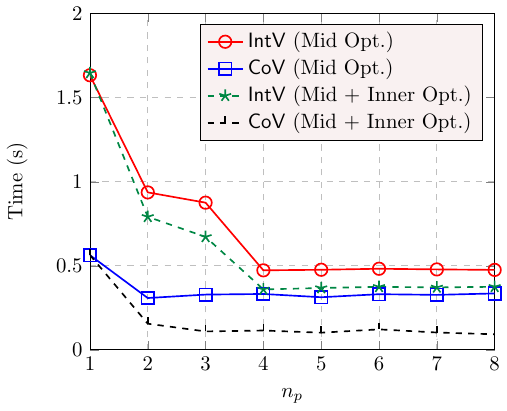}
        \caption{$\mathsf{IntV}$ and $\mathsf{CoV}$}
        \label{fig:bb1_bb2}
    \end{subfigure}
    
    \caption{Combined time performance comparison across different algorithms and parallel processing units ($n_p$) for $\lambda=128$.}
    \label{fig:combined}
\end{figure*}

\begin{figure}[htb]
    \centering
    \includegraphics[width=0.85\linewidth]{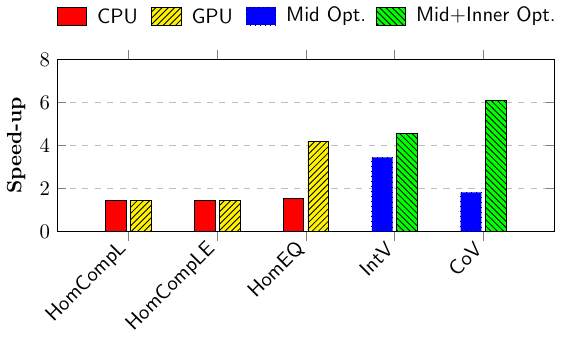}
    \caption{Speed-up summary for mid-level and innermost optimizations across core operations in \textsf{VeLoPIR}.}
    \label{fig:speedup_summary}
\end{figure}

\noindent \textbf{Innermost Optimization.} 
As shown in Fig.~\ref{fig:comple}, Fig.~\ref{fig:compl}, and Fig.~\ref{fig:equi}, the performance of \textsf{HomCompLE} and \textsf{HomCompL} is similar on both CPU and GPU, with minor variations across bit sizes. In contrast, \textsf{HomEQ} demonstrates a clear advantage on GPU, significantly outperforming CPU in both 16-bit and 32-bit operations. For \textsf{HomCompLE} and \textsf{HomCompL}, the maximum speed-up at $n_p = 8$ is modest (32-bit CPU/GPU: $1.43 \times$, 16-bit CPU: $1.41 \times$, 16-bit GPU: $1.35 \times$). For \textsf{HomEQ}, the speed-up on GPU is more pronounced, reaching $3.20 \times$ for 16-bit and $4.17 \times$ for 32-bit.

All three functions (\textsf{HomCompLE}, \textsf{HomCompL}, and \textsf{HomEQ}) reach optimal performance at $n_p = 8$, beyond which further increases in parallel processing units yield diminishing returns. This is attributed to overheads such as communication between processing units. While the theoretical complexity suggests a linear increase in speed-up with additional processing units, practical limitations, including synchronization and communication costs, prevent this from being realized. Thus, allocating 8 processing units, with a preference for GPU in the case of \textsf{HomEQ}, provides the best balance between performance and resource efficiency.

\noindent \textbf{Mid-Level Optimization: Validation Modules.} 
In our experiment, we first applied optimization at the mid-level of the computation hierarchy, focusing on the operational modes \textsf{IntV} and \textsf{CoV} (denoted as \textsf{IntV} (Mid Opt.) and \textsf{CoV} (Mid Opt.) in Fig.~\ref{fig:bb1_bb2}). In this setting, parallel processing units were allocated to each validation instance. For \textsf{IntV}, the maximum speed-up was observed at $n_p = 4$, achieving a $3.45\times$ improvement over the non-optimized baseline. We further applied inner-level optimization—targeting homomorphic comparison functions such as \(\mathsf{HomCompLE}\)—in a GPU-accelerated configuration referred to as \textsf{IntV} (Mid + Inner Opt.) in Fig.~\ref{fig:bb1_bb2}. This setting achieved a maximum speed-up of $4.54\times$ using $n_p = 32$ processing units. These results demonstrate that \textsf{IntV} benefits from both mid-level and innermost-level optimizations, though gains diminish beyond $n_p = 4$.

For \textsf{CoV}, the maximum speed-up was obtained at $n_p = 2$, achieving a $1.83\times$ improvement over the non-optimized baseline. Similar to \textsf{IntV}, the GPU-accelerated version---denoted as \textsf{CoV} (Mid + Inner Opt.) in Fig.~\ref{fig:bb1_bb2}---yielded a significantly higher speed-up of $6.09\times$. This demonstrates the effectiveness of homomorphic equality operations (e.g., \(\mathsf{HomEQ}\)) when parallelized on GPU architectures. These results are consistent with our theoretical parallel resource complexity summarized in Table~\ref{tab:hardware_complexity}, where \(d = 4\) for \textsf{IntV} and \(d = 2\) for \textsf{CoV}, corresponding to the number of parallel validation steps.

The observed speed-ups suggest that our theoretical model closely aligns with the experimental findings. Based on this evaluation, the optimal number of processing units was determined to be $n_p = 4$ for \textsf{IntV} and $n_p = 2$ for \textsf{CoV}, with additional performance gains enabled by innermost-level optimizations in the GPU-enhanced configurations. A summary of speed-ups across key operations—including \(\mathsf{HomCompLE}\), \(\mathsf{HomCompL}\), \(\mathsf{HomEQ}\), \textsf{IntV}, and \textsf{CoV}—is presented in Fig.~\ref{fig:speedup_summary}.

\begin{figure}[!htb]
    \centering
    \begin{subfigure}{0.23\textwidth}
        \centering
        \includegraphics[width=\textwidth]{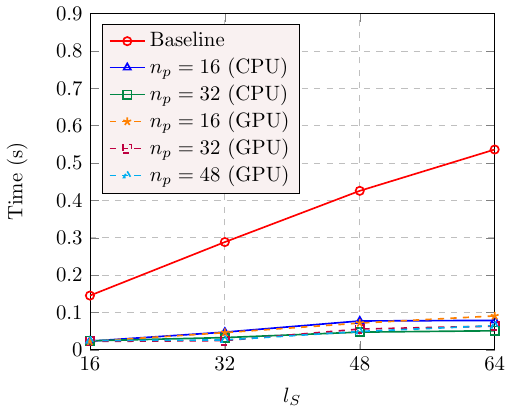}
        \caption{$\mathsf{HomBitwiseAND}$ by $l_S$}
        \label{fig:bit_AND}
    \end{subfigure}
    \hfill
    \begin{subfigure}{0.23\textwidth}
        \centering
        \includegraphics[width=\textwidth]{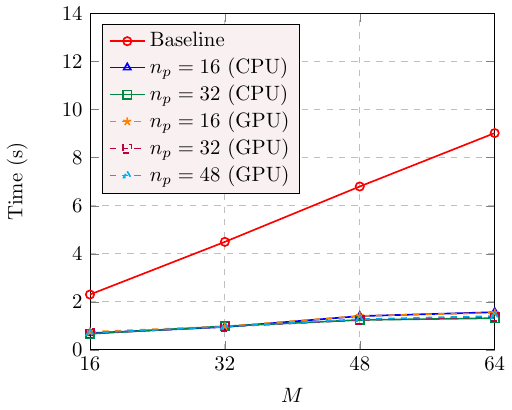}
        \caption{$\mathsf{HomSum}$ by $M$}
        \label{fig:HomSum}
    \end{subfigure}
    \caption{Time performance comparison of $\mathsf{HomBitwiseAND}$ and $\mathsf{HomSum}$ across different (a) service lengths ($l_S$) and (b) data sizes ($M$), with CPU and GPU optimizations.}
\end{figure}

\noindent \textbf{Mid-Level Optimization: HomBitwiseAND.} We evaluated the performance of the \textsf{HomBitwiseAND} algorithm by varying the service length (\(l_S\)) and the number of parallel processing units (\(n_p\)), including GPU acceleration. The results, as shown in Fig.~\ref{fig:bit_AND}, demonstrate a significant time performance improvement compared to the baseline. Both CPU and GPU optimizations achieve similar performance for smaller service lengths, particularly when \(n_p = 16\) or \(n_p = 32\), indicating that either approach can be used effectively depending on the available hardware. However, as \(n_p\) increases, we observed linear scaling of performance improvements up to the point where \(n_p \leq l_S\), beyond which further increases in parallel units provide diminishing returns. This is particularly evident with larger service lengths, where GPU acceleration initially shows superior performance compared to CPU, reaching a maximum speed-up of \(11.66 \times\) at \(l_S = 32\) (see Fig.~\ref{fig:speedup_bitwiseand}). Despite this, for even larger service lengths, such as \(l_S = 48\) and \(l_S = 64\), the benefits of increasing \(n_p\) on the GPU begin to decline. This decrease in performance is due to memory bandwidth limitations and data transfer overhead between the CPU and GPU. Thus, while GPU optimization remains preferable for larger workloads, the effectiveness is constrained by memory bottlenecks.

\begin{table*}[htb!]
\centering
\footnotesize
\caption{Execution time and speed-up across different optimization levels for various private information retrieval strategies using real-world datasets. The experiments evaluate bounding box validation (\textsf{IntV}), coordinate validation (\textsf{CoV}), and identifier matching (\textsf{IdM}) across four datasets.}
\begin{tabularx}{\textwidth}{|c|c|c|c|c|c|>{\centering\arraybackslash}X|}
\hline
\textbf{Dataset} & \textbf{Application Scenario} & \textbf{Method} & \textbf{Baseline (None)} & \textbf{Outermost} & \textbf{Outer + Mid Opt.} & \textbf{All} \\ \hline
\multirow{2}{*}{\texttt{covid-kor~\cite{covid-kor}}} 
  & Pandemic Alert & \textsf{IntV}  & $16.252$  & $3.304 (4.92 \times)$  & $2.783 (5.84 \times)$  & $2.512$ \textcolor{blue}{($6.47 \times$)}  \\ \cline{2-7}
  & City Matching & \textsf{IdM} & $4.44331$ & $1.47247 (3.02 \times)$ & $1.21494 (3.66 \times)$ & $1.13875$ \textcolor{blue}{($3.90 \times$)} \\ \hline
\multirow{2}{*}{\texttt{covid-usa}~\cite{covid-us}} 
  & Pandemic Alert & \textsf{IntV} & $113.44$  & $10.6831 (10.62 \times)$ & $10.1528 (11.17 \times)$ & $9.82627$ \textcolor{blue}{($11.55 \times$)} \\ \cline{2-7}
  & State Matching & \textsf{IdM} & $68.4045$ & $8.57437 (7.98 \times)$ & $7.8827 (8.68 \times)$ & $7.95777$ \textcolor{blue}{($8.59 \times$)} \\ \hline
\multirow{1}{*}{\texttt{gdacs}~\cite{gdacs}} 
  & Disaster Alert & \textsf{CoV} & $7.82023$ & $1.8384 (4.25 \times)$ & $1.70333 (4.59 \times)$ & $1.12841$ \textcolor{blue}{($6.93 \times$)} \\ \hline
\multirow{1}{*}{\texttt{weather-usa}~\cite{weather-usa}} 
  & Weather Info & \textsf{IdM} & $725.38$ & $67.3917 (10.76 \times)$ & $65.5437 (11.07 \times)$ & $65.7115$ \textcolor{blue}{($11.04 \times$)} \\ \hline
\end{tabularx}
\label{tab:execution_time_comparison}
\end{table*}

\noindent \textbf{Outermost Optimization: Aggregation via HomSum.} We assessed the performance of the \textsf{HomSum} algorithm by varying the number of data elements ($M$) and the number of parallel processing units (\(n_p\)), including GPU acceleration. As shown in Fig.~\ref{fig:HomSum}, both CPU and GPU implementations demonstrate significant performance gains compared to the baseline. Fig.~\ref{fig:speedup_homsum} further illustrates the speed-up achieved across different dataset sizes (\(M\)). As the number of data elements increases (from \(M = 16\) to \(M = 64\)), the speed-up remains consistent, with up to \(6.85 \times\) improvement on CPU with \(n_p = 32\). This is primarily due to the parallel reduction technique used in \textsf{HomSum}, which reduces the dataset by half with each iteration, enabling efficient scaling with \(M\). The number of processing units (\(n_p\)) was sufficient for the dataset sizes tested, avoiding any resource bottlenecks. Additionally, CPU and GPU performance are closely matched across all dataset sizes. For smaller datasets (\(M = 16\), \(M = 32\)), the CPU slightly outperforms the GPU, while for larger datasets (\(M = 48\), \(M = 64\)), the GPU has a marginal edge. However, the differences are minimal, demonstrating that both CPU and GPU offer comparable levels of acceleration for this task.
\begin{figure}[htb]
    \centering
    \begin{subfigure}[t]{0.23\textwidth}
        \centering
        \includegraphics[width=\textwidth]{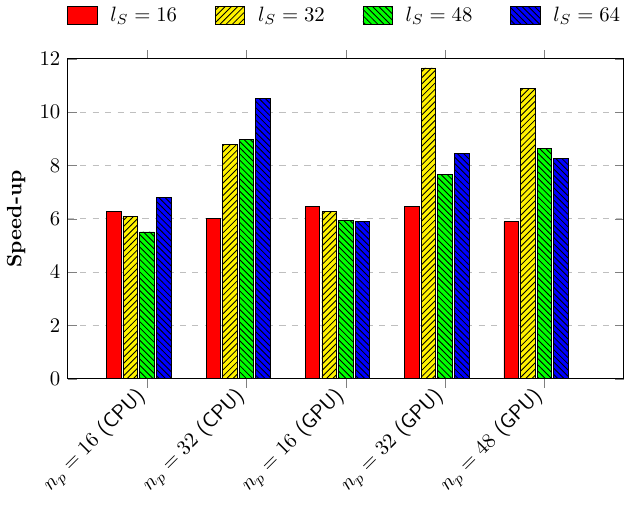}
        \caption{$\mathsf{HomBitwiseAND}$}
        \label{fig:speedup_bitwiseand}
    \end{subfigure}
    \hfill
    \begin{subfigure}[t]{0.23\textwidth}
        \centering
        \includegraphics[width=\textwidth]{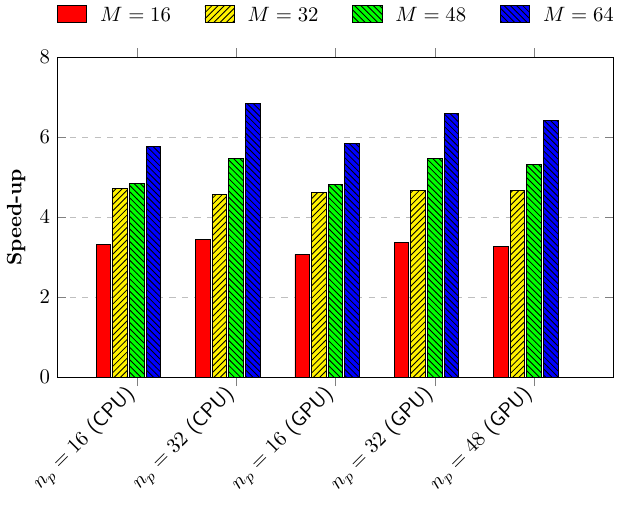}
        \caption{$\mathsf{HomSum}$}
        \label{fig:speedup_homsum}
    \end{subfigure}
    
    \caption{Speed-up comparison of different (a) service lengths ($l_S$) and (b) data sizes ($M$) with CPU or GPU optimization.}
\end{figure}

\subsection{Performance and Efficiency Analysis (RQ2)}

\noindent \textbf{Optimization Strategy.} In our evaluation, we implement three levels of optimization to enhance the efficiency of \textsf{VeLoPIR}, corresponding to the system's computational hierarchy (see Table~\ref{tab:hardware_complexity} for details). Outermost optimization targets large-scale processing and aggregation steps, parallelizing operations across records using GPU resources. Mid-level optimization focuses on the validation modes \textsf{IntV} and \textsf{CoV}, as well as the \textsf{HomBitwiseAND} operation, leveraging CPU and GPU parallelism where appropriate. Innermost optimization targets fundamental FHE comparison primitives, such as \(\mathsf{HomCompLE}\), \(\mathsf{HomCompL}\), and \(\mathsf{HomEQ}\), which are accelerated using GPU to maximize bit-level parallelism.

\noindent \textbf{{VeLoPIR} ({IntV}) Efficiency Compared to Baseline.} From Table~\ref{tab:execution_time_comparison}, we observe significant performance improvements achieved by our optimized \textsf{VeLoPIR} over the baseline.

For the \texttt{covid-kor} dataset, used in~\cite{yoo-torus}, we evaluated \textsf{VeLoPIR} using all three optimization levels under the \textsf{IntV} mode, consistent with the baseline setup. Our optimized version reduced the execution time from 16.252 seconds to 2.512 seconds, resulting in a $6.47\times$ overall speed-up (see Fig.~\ref{fig:speedup_bb1}). The most significant contribution came from outermost optimization, yielding a $4.92\times$ speed-up, while mid-level and innermost optimizations provided additional, though smaller, improvements.

For the \texttt{covid-usa} dataset, which is significantly larger (with 58 entries compared to 9 in \texttt{covid-kor}), outermost optimization again showed the most impact, achieving a speed-up of $10.62\times$. The larger dataset size is well-suited for GPU resource allocation at the outer loop level, enabling efficient parallel reduction across multiple records. However, for this dataset, the gains from mid-level and innermost optimizations diminished, likely due to overhead from CPU-GPU memory transmission as the dataset size and service length increased. This contrast between the \texttt{covid-kor} and \texttt{covid-usa} datasets highlights the scalability of outermost optimization, while also suggesting potential bottlenecks in deeper optimizations for larger inputs (see the comparison of mid-level and innermost speed-ups between \texttt{covid-kor} and \texttt{covid-usa} in Fig.~\ref{fig:speedup_bb1}).

\begin{figure}[htb]
    \centering
    \begin{subfigure}[t]{0.21\textwidth}
        \centering
        \includegraphics[width=\textwidth]{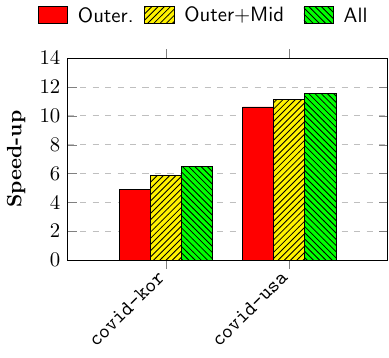}
        \caption{\textsf{VeLoPIR} using \textsf{IntV}}
        \label{fig:speedup_bb1}
    \end{subfigure}
    \hfill
    \begin{subfigure}[t]{0.25\textwidth}
        \centering
        \includegraphics[width=\textwidth]{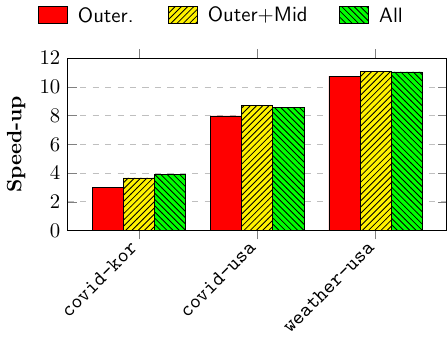}
        \caption{\textsf{VeLoPIR} using \textsf{IdM}}
        \label{fig:speedup_bb2+}
    \end{subfigure}
    
    \caption{Speed-up comparison across different optimization levels in \textsf{VeLoPIR}.}
    \label{fig:speedup_comparison}
\end{figure}

\noindent \textbf{Efficiency Evaluation of {VeLoPIR} ({IdM}).} 
We conducted experiments using the \textsf{IdM} mode across three datasets---\texttt{covid-kor}, \texttt{covid-usa}, and \texttt{weather-usa}---to evaluate the scalability of our optimization strategy with respect to dataset size (\(M\)) and service length (\(l_S\)). The speed-up results are illustrated in Fig.~\ref{fig:speedup_bb2+}, with detailed execution times provided in Table~\ref{tab:execution_time_comparison}.

As the dataset size increases, the overall speed-up also improves. For the smaller \texttt{covid-kor} dataset, we achieved a speed-up of $3.90\times$, while the \texttt{covid-usa} dataset yielded $8.59\times$. The largest dataset, \texttt{weather-usa}, achieved the highest speed-up of $11.04\times$, demonstrating the scalability of the \textsf{IdM} mode under our optimization framework.

Notably, the \texttt{weather-usa} dataset (\(M = 304\), \(l_S = 128\)) underscores the benefit of large-scale parallelism, especially at the outermost and mid-level optimization levels. However, we observed diminishing returns from innermost optimization as the dataset size increased. For instance, in the \texttt{covid-kor} dataset, innermost optimization contributed a modest gain (from $3.66\times$ to $3.90\times$), whereas for \texttt{covid-usa} and \texttt{weather-usa}, the speed-up slightly declined (from $8.68\times$ to $8.59\times$ and from $11.07\times$ to $11.04\times$, respectively). This suggests that GPU-accelerated innermost operations can become bottlenecked by memory bandwidth limitations at scale. Addressing this issue would require improved hardware support for more efficient data transfer between CPU and GPU.

\noindent \textbf{Efficiency Evaluation of {VeLoPIR} ({CoV}).} We conducted experiments using the \textsf{CoV} mode with the \texttt{gdacs} dataset (\(M = 9\), \(l_S = 16\)). A total speed-up of $6.93\times$ was achieved, primarily driven by outermost optimization, which contributed $4.25\times$. Additional gains were obtained through mid-level and innermost optimizations, with innermost optimization alone yielding a $2.34\times$ improvement via parallelized \(\mathsf{HomEQ}\) operations. Notably, no significant memory bandwidth issues or diminishing returns were observed in this setting, likely due to the relatively small dataset size and balanced CPU-GPU data transfer overhead.

\noindent \textbf{Scalability in Practical Settings.} While \textsf{VeLoPIR}'s design allows for high degrees of parallelism, its real---world scalability is affected by memory bandwidth constraints---particularly between CPU and GPU. As observed in Fig.~\ref{fig:bit_AND}, the performance of \(\mathsf{HomBitwiseAND}\), a mid-level optimization component, begins to saturate as the number of parallel units \(n_p\) increases under a fixed service length \(l_S\). A similar saturation trend is seen in Fig.~\ref{fig:HomSum} for \(\mathsf{HomSum}\), which belongs to the outermost optimization layer, as the dataset size \(M\) is held constant. These effects result in a divergence between theoretical and empirical speed-up, as summarized in Table~\ref{tab:execution_time_comparison}. 

In contrast, more lightweight primitive operations such as \(\mathsf{HomCompLE}\) and \(\mathsf{HomEQ}\), shown in Fig.~\ref{fig:combined}, continue to follow expected scaling behavior, particularly for modest parallelism levels (\(n_p = 4\) or \(8\)). This is because these primitives are less dependent on high-throughput memory access and thus less affected by interconnect bottlenecks. These results suggest that for large-scale deployments of \textsf{VeLoPIR}, improvements in memory bandwidth---especially between host and device---would allow the system to more closely approach its theoretical parallel performance bounds.

\subsection{Operational Mode Applicability (RQ3)}

\noindent \textbf{Information Alerts.} In scenarios such as pandemic alerts (\texttt{covid-kor} and \texttt{covid-usa} datasets), both \textsf{IntV} and \textsf{IdM} modes can be applied. However, \textsf{IdM} is generally more efficient, as it directly compares encrypted locational identifiers (e.g., city or state names) using homomorphic equality (\(\mathsf{HomEQ}\)), thereby avoiding bounding box computations and reducing circuit depth.

This performance advantage is evident in our results (see Fig.~\ref{fig:bb1_bb2+_application} and Fig.~\ref{fig:bb1_bb2+_opt_speedup}). For the \texttt{covid-kor} dataset, the baseline execution time for \textsf{IntV} was 16.252 seconds, which our optimized version reduced to 2.512 seconds—a $6.47\times$ speed-up. In contrast, \textsf{IdM} reduced the baseline of 4.443 seconds to 1.138 seconds, yielding a $3.90\times$ improvement.

The scalability of these modes is further highlighted by the results on the \texttt{covid-usa} dataset. Given its larger size, \textsf{IntV} achieved an $11.55\times$ speed-up (from 113.44 to 9.826 seconds), while \textsf{IdM} reduced the baseline from 68.404 to 7.957 seconds, resulting in an $8.59\times$ improvement.

\begin{figure}[htb]
    \centering
    \begin{subfigure}[t]{0.23\textwidth}
        \centering
        \includegraphics[width=\textwidth]{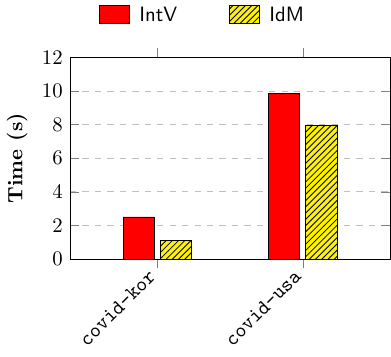}
        \caption{Execution Time (\textsf{IntV} vs. \textsf{IdM})}
        \label{fig:bb1_bb2+_application}
    \end{subfigure}
    \hfill
    \begin{subfigure}[t]{0.23\textwidth}
        \centering
        \includegraphics[width=\textwidth]{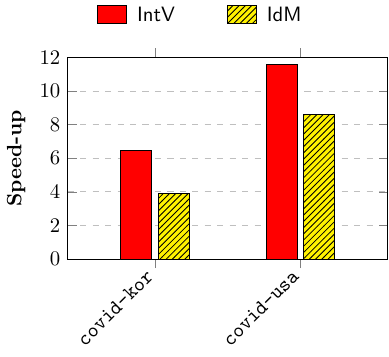}
        \caption{Speed-up (\textsf{IntV} vs. \textsf{IdM})}
        \label{fig:bb1_bb2+_opt_speedup}
    \end{subfigure}
    \caption{Time performance and speed-up comparison for \textsf{IntV} and \textsf{IdM} modes using the \texttt{covid-kor} and \texttt{covid-usa} datasets (information alert scenario).}
    \label{fig:bb1_bb2+_results}
\end{figure}

\noindent \textbf{Emergency Alerts.} The \textsf{CoV} mode is well-suited for emergency alert scenarios, as demonstrated by our experiments using the real-world \texttt{gdacs} dataset (Global Disaster Alert and Coordination System). Although relatively small in size (\(M = 9\), \(l_S = 16\)), the dataset represents realistic use cases such as satellite-based disaster monitoring. Its structure aligns well with the targeted nature of \textsf{CoV}, which enables exact coordinate-based validation for timely alerts.

The speed-up results indicate that \textsf{CoV} is highly efficient, reducing retrieval time to just 1.1284 seconds---a $6.93\times$ improvement over the baseline. This demonstrates the effectiveness of \textsf{CoV} in rapidly and accurately delivering alerts for specific geographic locations in disaster scenarios, without the overhead associated with more general-purpose modes like \textsf{IntV} or \textsf{IdM}.

Moreover, the \textsf{CoV} mode is broadly adaptable to other emergency use cases, such as crime alerts or evacuation warnings, where users in predefined coordinates require precise and immediate notifications.

\section{Conclusion}

We introduced \textsf{VeLoPIR}, a fast and versatile location-based PIR system that protects user locational privacy in real-world information and emergency alert scenarios. Through a set of flexible operational modes, \textsf{VeLoPIR} efficiently handles diverse datasets while preserving privacy. Performance evaluations on real-world datasets demonstrated significant speed-ups through parallel processing on both CPU and GPU, confirming the system's scalability. We also provided formal correctness proofs along with comprehensive security and privacy analyses, validating \textsf{VeLoPIR}'s effectiveness and robustness in practical deployments.




\bibliographystyle{plain}

\begin{thebibliography}{00}

\bibitem{nyt2018} S. Thompson and C. Warzel, ``Your apps know where you were last night, and they're not keeping it secret,'' \emph{The New York Times}, Dec. 10, 2018. [Online]. Available: \url{https://www.nytimes.com/interactive/2018/12/10/business/location-data-privacy-apps.html}


\bibitem{gravy2025} Z. Whittaker, ``A breach of Gravy Analytics’ huge trove of location data threatens the privacy of millions,'' \emph{TechCrunch}, Jan. 13, 2025. [Online]. Available: \url{https://techcrunch.com/2025/01/13/gravy-analytics-data-broker-breach-trove-of-location-data-threatens-privacy-millions/}


\bibitem{google2025} ``Google will pay Texas \$1.4B to settle claims the company collected users' data without permission,'' \emph{AP News}, May 9, 2025. [Online]. Available: \url{https://apnews.com/article/8097e181cc7cb8522781db8a9a897eea}

\bibitem{diffie2022} W. Diffie and M. E. Hellman, ``New directions in cryptography,'' in \emph{Democratizing Cryptography: The Work of Whitfield Diffie and Martin Hellman}, pp. 365--390, 2022.

\bibitem{menezes2018} A. J. Menezes, P. C. van Oorschot, and S. A. Vanstone, \emph{Handbook of Applied Cryptography}, Boca Raton, FL, USA: CRC Press, 2018.


\bibitem{he} R. L. Rivest, L. Adleman, and M. L. Dertouzos, ``On data banks and privacy homomorphisms,'' \emph{Found. Secure Comput.}, vol. 4, no. 11, pp. 169--180, 1978.

\bibitem{pir_gen} B. Chor, E. Kushilevitz, O. Goldreich, and M. Sudan, ``Private information retrieval,'' \emph{J. ACM}, vol. 45, no. 6, pp. 965--981, 1998.

\bibitem{pir_sup} A. Hamlin, R. Ostrovsky, M. Weiss, and D. Wichs, ``Private anonymous data access,'' in \emph{Advances in Cryptology--EUROCRYPT 2019: 38th Annual International Conference on the Theory and Applications of Cryptographic Techniques}, Darmstadt, Germany, May 19--23, 2019, pp. 244--273.

\bibitem{pir_sup2} P. Ananth, K.-M. Chung, X. Fan, and L. Qian, ``Collusion-resistant functional encryption for RAMs,'' in \emph{Int. Conf. Theory and Application of Cryptology and Information Security}, 2022, pp. 160--194.

\bibitem{pir_pre_1} A. Beimel, Y. Ishai, and T. Malkin, ``Reducing the servers' computation in private information retrieval: PIR with preprocessing,'' in \emph{Advances in Cryptology--CRYPTO 2000: 20th Annual Int. Cryptol. Conf.}, Santa Barbara, CA, USA, Aug. 20--24, 2000, pp. 55--73.

\bibitem{pir_pre_2} E. Boyle, Y. Ishai, R. Pass, and M. Wootters, ``Can we access a database both locally and privately?,'' in \emph{Theory of Cryptography: 15th Int. Conf., TCC 2017}, Baltimore, MD, USA, Nov. 12--15, 2017, pp. 662--693.

\bibitem{tfhe-1} I. Chillotti, N. Gama, M. Georgieva, and M. Izabachène, ``Faster fully homomorphic encryption: Bootstrapping in less than 0.1 seconds,'' in \emph{Advances in Cryptology--ASIACRYPT 2016: 22nd International Conference on the Theory and Application of Cryptology and Information Security}, Hanoi, Vietnam, Dec. 4--8, 2016, pp. 3--33.

\bibitem{tfhe-2} I. Chillotti, N. Gama, M. Georgieva, and M. Izabachène, ``TFHE: Fast fully homomorphic encryption over the torus,'' \emph{J. Cryptol.}, vol. 33, no. 1, pp. 34--91, 2020.



\bibitem{jain-two} M. Jain, P. Singh, and B. Raman, ``SHELBRS: Location-based recommendation services using switchable homomorphic encryption,'' in \emph{Int. Conf. Security, Privacy, and Appl. Cryptogr. Eng.}, 2021, pp. 63--80.

\bibitem{yoo-torus} J. S. Yoo, M. Y. Hong, J. W. Heo, K. H. Lee, and J. W. Yoon, ``Fast private location-based information retrieval over the torus,'' in \emph{Proc. 2024 IEEE Int. Conf. Adv. Video Signal Based Surveillance (AVSS)}, 2024, pp. 1--7.

\bibitem{pir_offline} H. Corrigan-Gibbs, A. Henzinger, and D. Kogan, ``Single-server private information retrieval with sublinear amortized time,'' in \emph{Advances in Cryptology--EUROCRYPT 2022: 41st Annual Int. Conf. Theory and Applications of Cryptographic Techniques}, 2022, pp. 3--33.

\bibitem{lin-double} W.-K. Lin, E. Mook, and D. Wichs, ``Doubly efficient private information retrieval and fully homomorphic RAM computation from ring LWE,'' in \emph{Proc. 55th Annu. ACM Symp. Theory of Comput. (STOC 2023)}, pp. 595--608, 2023.

\bibitem{bv} Z. Brakerski and V. Vaikuntanathan, ``Fully homomorphic encryption from ring-LWE and security for key dependent messages,'' in \emph{Advances in Cryptology--CRYPTO 2011: 31st Annu. Int. Cryptol. Conf.}, Santa Barbara, CA, USA, Aug. 14--18, 2011, pp. 505--524.




\bibitem{paillier} P. Paillier, ``Public-key cryptosystems based on composite degree residuosity classes,'' in \emph{Int. Conf. Theory and Appl. Cryptol. Techniques}, 1999, pp. 223--238.

\bibitem{elgamal} T. ElGamal, ``A public key cryptosystem and a signature scheme based on discrete logarithms,'' \emph{IEEE Trans. Inf. Theory}, vol. 31, no. 4, pp. 469--472, 1985.

\bibitem{an-covid} Y. An, S. Lee, S. Jung, H. Park, Y. Song, and T. Ko, ``Privacy-oriented technique for COVID-19 contact tracing (PROTECT) using homomorphic encryption: Design and development study,'' \emph{J. Med. Internet Res.}, vol. 23, no. 7, pp. e26371, 2021.

\bibitem{bfv} Z. Brakerski, ``Fully homomorphic encryption without modulus switching from classical GapSVP,'' in \emph{Annu. Cryptol. Conf.}, 2012, pp. 868--886.

\bibitem{ntru} J. Hoffstein, J. Pipher, and J. H. Silverman, ``NTRU: A ring-based public key cryptosystem,'' in \emph{Algorithmic Number Theory (ANTS III)}, 1998, pp. 267--288.

\bibitem{lwe-2009} O. Regev, ``On lattices, learning with errors, random linear codes, and cryptography,'' \emph{J. ACM}, vol. 56, no. 6, pp. 1--40, 2009.

\bibitem{lwe-2} D. Stehlé, R. Steinfeld, K. Tanaka, and K. Xagawa, ``Efficient public key encryption based on ideal lattices,'' in \emph{Int. Conf. Theory Appl. Cryptol. Inf. Security}, 2009, pp. 617--635.

\bibitem{gentry-fhe} C. Gentry, ``Fully homomorphic encryption using ideal lattices,'' in \emph{Proc. 41st Annu. ACM Symp. Theory Comput.}, 2009, pp. 169--178.

\bibitem{ckks} J. H. Cheon, A. Kim, M. Kim, and Y. Song, ``Homomorphic encryption for arithmetic of approximate numbers,'' in \emph{Advances in Cryptology--ASIACRYPT 2017}, Hong Kong, Dec. 3--7, 2017, pp. 409--437.

\bibitem{bgv} Z. Brakerski, C. Gentry, and V. Vaikuntanathan, ``(Leveled) fully homomorphic encryption without bootstrapping,'' \emph{ACM Trans. Comput. Theory}, vol. 6, no. 3, pp. 1--36, 2014.

\bibitem{spoof-1} N. O. Tippenhauer, C. Pöpper, K. B. Rasmussen, and S. Capkun, ``On the requirements for successful GPS spoofing attacks,'' in \emph{Proc. 18th ACM Conf. Comput. Commun. Security}, 2011, pp. 75--86.

\bibitem{spoof-2} K. C. Zeng, Y. Shu, S. Liu, Y. Dou, and Y. Yang, ``A practical GPS location spoofing attack in road navigation scenario,'' in \emph{Proc. 18th Int. Workshop Mobile Comput. Syst. Appl.}, 2017, pp. 85--90.

\bibitem{jam-1} Z. M. Kassas, J. Khalife, A. A. Abdallah, and C. Lee, ``I am not afraid of the GPS jammer: Resilient navigation via signals of opportunity in GPS-denied environments,'' \emph{IEEE Aerosp. Electron. Syst. Mag.}, vol. 37, no. 7, pp. 4--19, 2022.

\bibitem{jam-2} A. Grant, P. Williams, N. Ward, and S. Basker, ``GPS jamming and the impact on maritime navigation,'' \emph{J. Navig.}, vol. 62, no. 2, pp. 173--187, 2009.

\bibitem{pqc-survey} E. Zeydan, Y. Turk, B. Aksoy, and S. B. Ozturk, ``Recent advances in post-quantum cryptography for networks: A survey,'' in \emph{2022 7th Int. Conf. Mobile and Secure Services (MobiSecServ)}, 2022, pp. 1--8.

\bibitem{pqc-cyber} S. Paul, P. Scheible, and F. Wiemer, ``Towards post-quantum security for cyber-physical systems: Integrating PQC into industrial M2M communication,'' \emph{J. Comput. Security}, vol. 30, no. 4, pp. 623--653, 2022.

\bibitem{pqc-iot} Z. Liu, K.-K. R. Choo, and J. Grossschädl, ``Securing edge devices in the post-quantum Internet of Things using lattice-based cryptography,'' \emph{IEEE Commun. Mag.}, vol. 56, no. 2, pp. 158--162, 2018.

\bibitem{covid-us} E. Dong, H. Du, and L. Gardner, ``An interactive web-based dashboard to track COVID-19 in real time,'' \emph{Lancet Infect. Dis.}, vol. 20, no. 5, pp. 533--534, 2020, doi: 10.1016/S1473-3099(20)30120-1. [Online]. Available: \underline{https://github.com/CSSEGISandData/COVID-19}.

\bibitem{covid-kor} Korea Disease Control and Prevention Agency (KDCA), ``COVID-19 daily incidences in nine major cities in Korea,'' 2021. [Online]. Available: \underline{https://ncv.kdca.go.kr/pot/cv/trend/dmstc/selectMntrgSttus.do}.

\bibitem{gdacs} Global Disaster Alert and Coordination System (GDACS), ``Global Disaster Alert and Coordination System (GDACS) alerts and impact estimations,'' 2024. [Online]. Available: \underline{https://www.gdacs.org/Alerts/default.aspx}.

\bibitem{weather-usa} National Oceanic and Atmospheric Administration (NOAA), ``Daily weather reports for U.S. cities,'' 2016. [Online]. Available: \underline{https://www.noaa.gov/climate}.


\end{thebibliography}

\appendix
\appendix

\section{CORRECTNESS PROOF}
\subsection{Proof of Theorem 1}
\label{appendix:proof_LE}

\begin{proof}

The main part of the \(\mathsf{HomCompLE}\) algorithm relies on the combination of homomorphic gates, specifically the \(\mathsf{HomXNOR}\) gate followed by the \(\mathsf{HomMUX}\) gate, to evaluate the encrypted inputs.

First, the algorithm initializes \(t_2\) with a trivial TLWE ciphertext \(\mathsf{Enc}_{\mathbf{s}}^{\mathsf{TLWE}}(\mathbf{0}, \frac{1}{8})\), which is an encryption of \(1\) under the secret key \(\mathbf{s}\). This assumes \(\mathsf{ct_1} \leq \mathsf{ct_2}\) initially.

If \(\mathsf{ct_1}\) and \(\mathsf{ct_2}\) are equal, then for all bits, the \(\mathsf{HomXNOR}\) gate will output \(t_1 = \mathsf{Enc}_{\mathbf{s}}(1)\), preserving \(t_2 = \mathsf{Enc}_{\mathbf{s}}(1)\) throughout the loop. The final output, \(r\), becomes \(\mathsf{Enc}_{\mathbf{s}}(1)\), confirming \(\mathsf{ct_1} \leq \mathsf{ct_2}\) as expected.

Consider the highest bit \(j\) where \(\mathsf{ct_1}[j]\) and \(\mathsf{ct_2}[j]\) differ:
\begin{equation*}
\begin{aligned}
&\text{(Case 1): } \mathsf{ct_1}[j] = 1, \mathsf{ct_2}[j] = 0 \\
&\text{(Case 2): } \mathsf{ct_1}[j] = 0, \mathsf{ct_2}[j] = 1
\end{aligned}
\end{equation*}
\begin{itemize}
    \item Case 1. Here, \(\mathsf{HomXNOR}\) outputs \(t_1 = \mathsf{Enc}_{\mathbf{s}}(0)\). Consequently, \(\mathsf{HomMUX}\) selects \(\mathsf{ct_2}[j]\), setting \(t_2 = \mathsf{Enc}_{\mathbf{s}}(0)\). Thus, the result is \(\mathsf{Enc}_{\mathbf{s}}(0)\), correctly indicating that \(\mathsf{ct_1} > \mathsf{ct_2}\).
    \item Case 2. In this scenario, \(\mathsf{HomXNOR}\) outputs \(t_1 = \mathsf{Enc}_{\mathbf{s}}(0)\), leading \(\mathsf{HomMUX}\) to select \(\mathsf{ct_2}[j]\), thus \(t_2 = \mathsf{Enc}_{\mathbf{s}}(1)\). Therefore, the result is \(\mathsf{Enc}_{\mathbf{s}}(1)\), correctly showing that \(\mathsf{ct_1} < \mathsf{ct_2}\).
\end{itemize}

The \(\mathsf{HomMUX}\) at the end considers \(t_0\), which indicates differing signs between \(\mathsf{ct_1}\) and \(\mathsf{ct_2}\). If signs differ, \(t_0\) directs the selection to \(\mathsf{ct_1}[l_I - 1]\). This ensures correct handling of signed numbers:
\begin{itemize}
    \item If \(\mathsf{ct_1}\) is negative and \(\mathsf{ct_2}\) is positive, \(\mathsf{ct_1} < \mathsf{ct_2}\) holds, setting the result as \(\mathsf{Enc}_{\mathbf{s}}(1)\).
    \item Conversely, if \(\mathsf{ct_1}\) is positive and \(\mathsf{ct_2}\) is negative, \(\mathsf{ct_1} > \mathsf{ct_2}\), leading to \(\mathsf{Enc}_{\mathbf{s}}(0)\).
\end{itemize}

\end{proof}

\subsection{Proof of Corollary 1}
\label{appendix:proof_L}

\begin{proof}
The algorithm \(\mathsf{HomCompL}(\mathsf{ct_1}, \mathsf{ct_2}, \mathsf{evk})\) builds upon the logic established in Theorem~\ref{thm:comple} for the \(\mathsf{HomCompLE}\) algorithm. The key difference in \(\mathsf{HomCompL}\) lies in the initialization of the variable \(t_2\), which is set to \(\mathsf{Enc}_{\mathbf{s}}^{\mathsf{TLWE}}(\mathbf{0}, -\frac{1}{8})\), corresponding to an encryption of 0.

In the \(\mathsf{HomCompLE}\) algorithm, \(t_2\) was initialized to an encryption of 1, representing the assumption that \(\mathsf{ct_1} = \mathsf{ct_2}\). Here, by initializing \(t_2\) to 0, we instead assume that \(\mathsf{ct_1} \neq \mathsf{ct_2}\). The iterative loop in \(\mathsf{HomCompL}\) uses the \(\mathsf{HomXNOR}\) and \(\mathsf{HomMUX}\) gates to compare the bits of \(\mathsf{ct_1}\) and \(\mathsf{ct_2}\). If all bits of \(\mathsf{ct_1}\) and \(\mathsf{ct_2}\) are equal, \(t_2\) remains as 0, which is the correct result since \(\mathsf{ct_1} = \mathsf{ct_2}\) means \(\mathsf{ct_1} < \mathsf{ct_2}\) is false.

The variable \(t_0\), as in \(\mathsf{HomCompLE}\), accounts for cases where \(\mathsf{ct_1}\) and \(\mathsf{ct_2}\) have different signs. If the signs differ, \(t_0\) directs the output based on the sign of \(\mathsf{ct_1}\). This ensures that the algorithm correctly handles all cases where \(\mathsf{ct_1} < \mathsf{ct_2}\), including when the signs are different.

Thus, by initializing \(t_2\) to 0, \(\mathsf{HomCompL}\) ensures that when \(\mathsf{ct_1} = \mathsf{ct_2}\), the output is \(\mathsf{Enc}_{\mathbf{s}}(0)\), and when \(\mathsf{ct_1} < \mathsf{ct_2}\), the output is \(\mathsf{Enc}_{\mathbf{s}}(1)\). This proves that the algorithm correctly evaluates \(\mathsf{ct_1} < \mathsf{ct_2}\) as stated.
\end{proof}

\subsection{Proof of Lemma 1}
\label{appendix:proof_EQ}

\begin{proof}
The algorithm starts by initializing the result $r$ as a trivial $\mathsf{TLWE}$ encryption of 1, denoted as $\mathsf{Enc}_{\mathbf{s}}^{\mathsf{TLWE}}(0, \frac{1}{8})$ assuming that $\mathsf{ct_1} = \mathsf{ct_2}$ by default.

In each iteration of the loop, the $\mathsf{HomXNOR}$ gate compares corresponding bits of $\mathsf{ct_1}$ and $\mathsf{ct_2}$. If the bits are equal, $\mathsf{HomXNOR}$ outputs $\mathsf{Enc}_{\mathbf{s}}(1)$; otherwise, it outputs $\mathsf{Enc}_{\mathbf{s}}(0)$. The result $t$ from each bit comparison is then combined with the current value of $r$ using the $\mathsf{HomAND}$ operation.

The $\mathsf{HomAND}$ operation is crucial: it ensures that if any bit of $\mathsf{ct_1}$ and $\mathsf{ct_2}$ differs, the final result $r$ becomes $\mathsf{Enc}_{\mathbf{s}}(0)$, indicating inequality. Conversely, if all corresponding bits of $\mathsf{ct_1}$ and $\mathsf{ct_2}$ are identical, the final value of $r$ remains $\mathsf{Enc}_{\mathbf{s}}(1)$, confirming that $\mathsf{ct_1}$ and $\mathsf{ct_2}$ encrypt the same value.

Thus, the algorithm correctly evaluates the equality of $z_1$ and $z_2$ under homomorphic encryption, as claimed.
\end{proof}

\subsection{Proof of Theorem 2}
\label{appendix:proof_bb1}

\begin{proof}
The correctness of this algorithm can be verified by analyzing the behavior of the homomorphic comparison operations ($\mathsf{HomCompL}$ and $\mathsf{HomCompLE}$ operations). 

For the latitude, \(v_{x_{left}}\) and \(v_{x_{right}}\) will both be \(\mathsf{Enc}_{\mathbf{s}}(1)\) only when \(x_{left} \leq x < x_{right}\). Thus, the combined result $v_x = \mathsf{HomAND}(v_{x_{left}}, $ $v_{x_{right}}, \mathsf{evk})$ will be \(\mathsf{Enc}_{\mathbf{s}}(1)\) only if this condition holds. Conversely, consider a case where \(x < x_{left}\); in this scenario, \(v_{x_{left}}\) will be \(\mathsf{Enc}_{\mathbf{s}}(0)\), ensuring that \(v_x\) becomes \(\mathsf{Enc}_{\mathbf{s}}(0)\) irrespective of the value of \(v_{x_{right}}\).

Similarly, the combined result $v_y = \mathsf{HomAND}(v_{y_{left}}, v_{y_{right}}$ $, \mathsf{evk})$ will be \(\mathsf{Enc}_{\mathbf{s}}(1)\) only if \(y_{left} \leq y < y_{right}\) holds. 

The final validation \(v = \mathsf{HomAND}(v_x, v_y, \mathsf{evk})\) will output \(\mathsf{Enc}_{\mathbf{s}}(1)\) if both latitude and longitude conditions are satisfied, ensuring the coordinate \((x, y)\) is within the interval. Otherwise, it outputs \(\mathsf{Enc}_{\mathbf{s}}(0)\).
\end{proof}

\section{SECURITY PROOF}
\subsection{Proof of Corollary 2}
\label{appendix:proof_raw}

\begin{proof}
The server's goal is to retrieve the secret key $\mathbf{s}$ from the encrypted dataset $\mathsf{Enc}_{\mathsf{pk}}(\mathbf{D})$. This problem is essentially identical to the one discussed in Theorem~\ref{thm:aux_adv}, since the ciphertext $\mathsf{ct}$ in $\mathsf{Enc}_{\mathsf{pk}}(\mathbf{D})$ is of the form $\mathsf{ct} \leftarrow (\mathbf{0}, \mathsf{Ecd}(x)) + \mathsf{pk}$. Thus, the server must extract the secret key from $\mathsf{pk}$, which consists of up to $W = M  l_I d + M  l_S$ TLWE samples. Therefore, the server's advantage is bounded by:
\[
\text{Adv}_{\text{TLWE}_{1 \text{ to } W}}^{\textsf{SE}} < \frac{O(W)}{2^\lambda}
\]
where $W = M  l_I d + M  l_S$.
\end{proof}

\section{ADDITIONAL SUPPORTING INFORMATION}

In this section, we briefly explain the bootstrapping procedure in \textsf{TFHE}, which is critical for performing homomorphic gate operations. We also provide the \textsf{TFHE} parameters used during the evaluation of \textsf{VeLoPIR}, followed by a summary of the notations used throughout the paper.





\subsection{\textsf{TFHE} Parameters}
\label{appendix:tfhe_params}
Table~\ref{tab:tfhe_params} lists the cryptographic parameters for the TFHE scheme at the 128-bit security level, as used in \cite{tfhe-2}.

\begin{table}[!htb]
  \centering
  \caption{Cryptographic parameters for \textsf{TFHE} 128-bit security level}
  \label{tab:tfhe_params}
  \begin{tabularx}{\linewidth}{lcc}
    \toprule
    \multicolumn{2}{c}{\textbf{Parameter}}                                  & \textbf{Value} \\ 
    \midrule
    \textsf{TLWE} Dimension            & $n$                    & 630           \\ 
    \textsf{TRLWE} Dimension    & $k$                    & 1             \\ 
    Polynomial Size          & $N$                    & 1024          \\ 
    Key Switch Base Log      & $\log_2(\beta_{KS})$    & 2             \\ 
    Key Switch Level         & $\ell_{KS}$             & 8             \\ 
    Key Switch Standard Deviation & $\sigma_{KS}$      & $2^{-15}$     \\ 
    Bootstrapping Key Base Log & $\log_2(\beta_{BK})$  & 7             \\ 
    Bootstrapping Key Level  & $\ell_{BK}$             & 3             \\ 
    Bootstrapping Key Standard Deviation & $\sigma_{BK}$ & $2^{-25}$   \\ 
    \bottomrule
  \end{tabularx}
\end{table}




\end{document}